\documentclass[12pt]{article}
\newif\ifHideComments

\HideCommentsfalse 

\usepackage[utf8]{inputenc}
\usepackage[T1]{fontenc}
\usepackage[english]{babel}
\usepackage{microtype}
\usepackage{geometry}
\usepackage{setspace}
\usepackage{graphicx}
\usepackage{caption}
\usepackage{subcaption}
\usepackage{booktabs}
\usepackage[flushleft]{threeparttable}
\usepackage{tabularx}
\usepackage{ragged2e}
\usepackage{multirow}
\usepackage{rotating}
\usepackage{float}
\usepackage{thmtools}
\usepackage{amsthm}

\usepackage{amsmath,amssymb,amsthm}
\usepackage{dsfont}
\usepackage{bm}
\usepackage{optidef}
\usepackage[backref=page]{hyperref}
\usepackage[open,openlevel=2,atend,numbered,level=4]{bookmark}
\usepackage[nameinlink]{cleveref}
\usepackage[authoryear]{natbib}
\usepackage{enumitem}
\usepackage[dvipsnames]{xcolor}
\usepackage{tikz}
\usepackage{titling}
\usepackage{dirtytalk}
\usepackage{thm-restate}

\definecolor{deepblue}{RGB}{0,0,110}
\definecolor{darkgreen}{rgb}{0.0,0.5,0.0}

\renewcommand*{\backrefalt}[4]{%
    \ifcase #1
        No citation in the text.
    \else
        [#2]
    \fi}
\hypersetup{
    pdftitle={Allocating Students to Schools: Theory, Methods, and Empirical Insights},
    pdfauthor={Yeon-Koo Che, Julien Grenet, YingHua He},
    colorlinks=true,
    linkcolor={deepblue},
    citecolor={deepblue},
    urlcolor={deepblue},
    hypertexnames=false,
    pdfpagemode=UseNone,
    pdfdisplaydoctitle=true
}

\setlist[enumerate]{leftmargin=1.5cm,rightmargin=0.5cm,noitemsep, topsep=2pt}

\appto\TPTdoTablenotes{\labelsep0.0em\footnotesize}

\allowdisplaybreaks

\makeatletter
\newtheoremstyle{bluehead}%
  {}{}                
  {\itshape}          
  {}                  
  {\color{deepblue}\bfseries} 
  {.}                 
  { }                 
  {%
    \thmname{#1}\thmnumber{ #2}\thmnote{ \textbf{(#3)}}%
  }                   
\makeatother

\theoremstyle{bluehead}

\newtheorem{theorem}{Theorem}

\newtheorem{lemma}{Lemma}
\newtheorem{proposition}{Proposition}
\newtheorem{corollary}{Corollary}
\crefname{equation}{}{}
\crefname{theorem}{Theorem}{Theorems}
\crefname{figure}{Figure}{Learning theory & Prior set/Figures}
\crefname{table}{Table}{Tables}
\crefname{section}{Section}{Sections}

\def\half{\frac{1}{2}}

\newcommand{\E}{\mathbb{E}}

\ifHideComments
    
    \newcommand{\jg}[1]{} 
    \newcommand{\yk}[1]{}
    \newcommand{\ml}[1]{}

\else
    
    \newcommand{\jg}[1]{\textcolor{red}{[JG: #1]}} 
    \newcommand{\yk}[1]{\textcolor{blue}{[YK: #1]}}
    \newcommand{\ml}[1]{\textcolor{red}{[ML: #1]}}

\fi

\thanksmarkseries{arabic}

\makeatletter
\renewcommand*{\@fnsymbol}[1]{\ifcase#1\or *\else\@arabic{#1}\fi}
\makeatother

 \DeclareMathAlphabet{\pazocal}{OMS}{zplm}{m}{n}

\definecolor{clemson-orange}{RGB}{234,106,32}
\definecolor{chicago-maroon}{RGB}{128,0,0}
\definecolor{northwestern-purple}{RGB}{82,0,99}
\definecolor{cornell-red}{RGB}{179,27,27}
\definecolor{sauder-green}{RGB}{171,180,0}
\definecolor{gray}{RGB}{192,192,192}
\definecolor{lawngreen}{RGB}{0,250,154}
 
\usepackage{fullpage}

\usepackage[nameinlink]{cleveref}
\usepackage{graphicx}

\providecommand{\U}[1]{\protect\rule{.1in}{.1in}} \textheight 8.2in
\def\t{\theta}
\def\E{\mathbb{E}}

\newcommand{\Prb}{\Pr}
\newcommand{\tstate}{\t}
\newcommand{\pitrue}{\pi_{\text{true}}}
\theoremstyle{definition}

\newtheorem*{algorithm*}{Greedy Algorithm}

\definecolor{darkgreen}{rgb}{0.0, 0.5, 0.0}
\definecolor{brightpink}{rgb}{1.0, 0.0, 0.5}
\definecolor{brightgreen}{rgb}{0.4, 1.0, 0.0}
\definecolor{shockingpink}{rgb}{0.99, 0.06, 0.75}
\definecolor{persiangreen}{rgb}{0.0, 0.65, 0.58}

\def\bq{\begin{equation}}
\def\eq{\end{equation}}
\def\ba{\begin{eqnarray}}
\def\ea{\end{eqnarray}}
\def\bas{\begin{eqnarray*}}
\def\eas{\end{eqnarray*}}

\def\E{\mathbb{E}}

\def\t{\theta}

\def\U{\Upsilon}

\makeatletter
\def\mathcenterto#1#2{\mathclap{\phantom{#1}\mathclap{#2}}\phantom{#1}}
\let\old@widetilde\widetilde
\def\widetildeto#1#2{\mathcenterto{#2}{\old@widetilde{\mathcenterto{#1}{#2\,}}}}
\let\old@widehat\widehat
\def\widehatto#1#2{\mathcenterto{#2}{\old@widehat{\mathcenterto{#1}{#2\,}}}}
\makeatother

\definecolor{clemson-orange}{RGB}{234,106,32}
\definecolor{chicago-maroon}{RGB}{128,0,0}
\definecolor{northwestern-purple}{RGB}{82,0,99}
\definecolor{cornell-red}{RGB}{179,27,27}
\definecolor{sauder-green}{RGB}{171,180,0}
\definecolor{gray}{RGB}{192,192,192}
\definecolor{lawngreen}{RGB}{0,250,154}

\def\bp{\bm{\pi}}

\allowdisplaybreaks

\title{Learning Against Nature: Minimax Regret and the Price of Robustness}

\author{
    Yeon-Koo Che\thanks{Department of Economics, Columbia University, USA. Email: \href{mailto:yeonkooche@gmail.com}{\texttt{yeonkooche@gmail.com}}.}\qquad
    Longjian Li\thanks{Department of Economics, New York University, USA. Email: \href{mailto:ll4830@nyu.edu}{\texttt{ll4830@nyu.edu}}.}\qquad
    Tianling Luo\thanks{Department of Economics, Columbia University, USA. Email: \href{mailto:tl3078@columbia.edu}{\texttt{tl3078@columbia.edu}}.}
}

\date{\today}

\begin{document}

\setstretch{1.2} 

\thispagestyle{empty}

\maketitle

\begin{abstract}\singlespacing%
\noindent  
We study how a decision-maker (DM) learns from data of unknown quality to form robust, ``general-purpose'' posterior beliefs. We develop a framework for robust learning and belief formation under a minimax-regret criterion, cast as a zero-sum game: the DM chooses posterior beliefs to minimize ex-ante regret, while an adversarial Nature selects the data-generating process (DGP). We show that, in large samples of $n$ signal draws, Nature optimally induces ambiguity by choosing a process whose precision converges to the uninformative signals at the rate $1/\sqrt{n}$. As a result, learning against the adversarial DGP is nontrivial as well as {\it incomplete}: the DM's ex-ante regret remains strictly positive even with an infinite amount of data. However, {\it when the true DGP is fixed and informative (even if only slightly)}, our DM with a robust updating rule eventually learns the state with enough data.  Still, learning occurs at a sub-exponential rate---quantifying the asymptotic price of robustness---and it exhibits ``under-inference'' bias.  Our framework provides a {\it decision-theoretic  dual} to the {\it local alternatives method} in asymptotic statistics, deriving the characteristic $1/\sqrt{n}$-scaling endogenously from the signal ambiguity.
\bigskip

\end{abstract}




\section{Introduction}  

Learning is fundamentally a process of interpreting  data to infer an unknown state of interest. The standard economic approach is Bayesian, assuming a decision-maker (DM) who utilizes an unambiguous model of the data-generating process (DGP). In practice, however, decision-makers often face ambiguity regarding the true DGP, which complicates how they draw inferences from observed data. This paper proposes a model of how a DM may learn robustly---i.e., form posterior beliefs---when the process generating her data is uncertain.

Specifically, we consider a decision-maker (DM) facing an unknown binary state $\theta \in \{0, 1\}$. To learn about the state, she observes a sample of $n$ conditionally i.i.d. signals. The core friction is that the DM faces ambiguity regarding the data-generating process (DGP)—she knows the signals carry information but does not know their precise quality.  The DM resolves this ambiguity using a {\it minimax regret} criterion: she selects a posterior belief $a \in [0,1]$ to minimize the maximum expected regret—--the loss resulting from the divergence between her belief and the ``Oracle'' belief she would have held had she known the true DGP.\footnote{\label{fn:maxmin}  We favor minimax regret over a maxmin utility criterion to avoid the triviality of the latter. Under a maxmin criterion, an adversarial Nature would select a completely uninformative DGP, leading the DM to simply ignore the data. Minimax regret, by contrast, captures the desire to guard against the ignorance of the DGP; it penalizes the DM for ignoring data that could have been significantly informative, thus justifying a meaningful engagement with the signal.}

Previous literature typically examines the worst-case belief for a particular decision problem after observing a specific signal realization. However, under standard linear preferences, these ``interim'' or decision-specific approaches yield beliefs that depend sensitively on the specific payoff structure and are often extreme—assigning all probability to a single state except for knife-edge signal realizations. We depart from this tradition by considering an \textit{ex-ante} rule that the DM uses consistently to interpret any signal realization across a spectrum of potential decision problems. This approach is motivated both by {\it cognitive economy}:  it may be cognitively burdensome to form a new interpretive model for every task one encounters, and by a {\it consistency requirement}: while Nature is adversarial, it is a single entity that cannot strategically vary the underlying DGP across different decision contexts and signal realizations.

Following the logic of \cite{savage1971elicitation} and \cite{schervish1989general}, we argue that the regret for such a DM can be characterized by an expected Bregman divergence. Qualitatively, this divergence represents the distance between the DM's belief and the Oracle's belief. This formulation induces a strictly convex ex-ante regret function that is general enough to encompass standard statistical distances such as Mean Squared Error and Kullback-Leibler divergence. Crucially, this convexity ensures that the optimal robust rule is characterized by stable, non-degenerate beliefs. By aggregating over a continuum of potential tasks, the DM arrives at a rule that provides a stable interpretation of data, avoiding the all-or-nothing pathologies inherent in decision-dependent models.

We utilize the minimax theorem to recast this problem as a zero-sum game between the DM and an adversarial Nature. To build intuition, we provide a precise characterization of the worst-case DGP for the specialized case of binary signals and MSE regret. In the large-sample limit, these DGPs converge in log-odds ratio to (mixed) Gaussian limit experiments. Nature's optimal strategy is to randomize between a totally uninformative signal and a weakly informative signal that approaches the uninformative boundary at the rate of $1/\sqrt{n}$.

Under this ``strategic scaling,'' the uninformative signal converges (again in log odds ratio) to a standard normal distribution $\mathcal{N}(0, 1)$. The informative signal converges to $\mathcal{N}(2c, 1)$ if $\theta=1$ and $\mathcal{N}(-2c, 1)$ if $\theta=0$ for an optimally chosen constant $c > 0$. Nature chooses $c$ and the randomization weight to maximize the DM's ambiguity, ensuring the learning problem remains non-trivial even as $n \to \infty$.  The DM then forms her posterior belief over the resulting mixture-normal distribution.

We next show that this $1/\sqrt{n}$ rate is not an artifact of binary signals or MSE. Our General Learning Result proves that for any arbitrary signal space and any Bregman divergence, Nature must collapse to the uninformative experiment at exactly the $1/\sqrt{n}$ rate to sustain a non-trivial equilibrium. This provides a rigorous decision-theoretic interpretation for the ``local alternatives'' framework in statistics.

As an endogenous outcome of this game, guaranteed learning against the adversarial Nature is incomplete. Even with infinite data, the DM’s ex-ante regret remains strictly positive against the worst-case DGPs.  At the same time, if the true DGP is fixed and informative (even if only slightly), then our DM eventually learns the state with enough data.  However, her error vanishes only at a sub-exponential rate, specifically $\Theta(e^{-C\sqrt{n}})$. This gap---relative to the exponential rate of an Oracle---quantifies the asymptotic price of robustness. Behaviorally, our DM may misinterpret data in either direction for a finite $n$.  However, for sufficiently large $n$, our DM always underinfers from the data. Because she must guard against a worst-case scenario in which Nature has degraded the signal's precision to nearly uninformative levels, she regards the observed evidence as systematically weaker than its true quality. Essentially, the robust DM remains skeptical that the signal is truly informative, leading to an ``asymptotic caution'' that slows the learning process significantly compared to the Bayesian benchmark.


Beyond characterizing robust beliefs, our findings offer some insight for asymptotic statistics. The $1/\sqrt{n}$ scaling---often employed as a technical device in asymptotic statistics---emerges here as a fundamental equilibrium property. In standard models with fixed signal precision, the Law of Large Numbers eventually renders global uncertainty trivial. To maintain a non-degenerate analysis in large samples, researchers typically ``zoom in''  on local alternatives that are microscopic distances apart.  

In our framework, the non-vanishing uncertainty is not an assumption, but a consequence of strategic ambiguity. Nature degrades the signal precision at the $1/\sqrt{n}$ rate specifically to counteract the growing sample size, ensuring that the ``fog'' of the limit experiment never lifts. This provides a decision-theoretic dual to the local alternatives method: rather than scaling the state space to match the data, we show that an adversarial Nature scales the precision of the data to match the sample size.

\section{Model}\label{sec:model}
\subsection{Primitives}
We consider a single-agent learning problem regarding a binary state of the world, $\theta \in \{0, 1\}$. The decision-maker (DM) holds a general prior $\mu = \Pr(\theta=1) \in (0,1)$. To learn about the state, the DM observes data generated by an \emph{experiment} (or Data-Generating Process (DGP)), denoted by $\bp = (\pi_1, \pi_0)$. Here, $\pi_\theta \in \Delta(\mathcal{S})$ specifies the probability distribution over a finite set $\mathcal{S}$ of signal realizations $s$  conditional on the state $\theta$.

The DM observes a vector of signal counts $K = (K_s)_{s \in \mathcal{S}}$ generated from $n$ i.i.d. trials of the experiment. The central friction of the model is that the DM faces \emph{ambiguity} about the true data-generating process. She knows only that $\bp$ belongs to a set of feasible experiments $\Pi \subseteq \Delta(\mathcal{S})^{\{0,1\}}$. We assume $\Pi$ is compact and contains a unique uninformative experiment $\pi^0$ (where $\pi_1 = \pi_0$). Because the exact properties of the signal are unknown, the DM cannot simply apply Bayes' rule to form a standard posterior.

\subsection{The Robust Learning Problem}

The DM's objective is to form a {\bf belief-formation rule} that maps each signal count vector $K$ she may observe to the posterior belief---an ``action''---$a(K) \in [0,1]$, representing the probability that $\theta=1$. The DM evaluates the quality of this belief-formation rule by comparing it against a benchmark rule: an ``Oracle'' who knows the true experiment $\bp$ and forms the correct Bayesian posterior $q_{\text{oracle}}(K; \bp)$ for each $K$. The DM then chooses the rule that minimizes her regret relative to this benchmark in the worst-case DGP.

Formally, let $R(q_{\text{oracle}}, a)$ denote the ex-ante regret (defined in the sequel) of choosing beliefs $a$ when the true rational beliefs are $q_{\text{oracle}}$. The DM solves the following minimax problem:\footnote{Recall \Cref{fn:maxmin} for the explanation on why we have adopted minimax regret criteria over maxmin utility criteria. }
\begin{equation*}
    \min_{a(\cdot)} \max_{\sigma_n \in \Delta(\Pi)} \mathbb{E}_{\sigma_n, K} \left[ R\left(q_{\text{oracle}}(K; \bp), a(K)\right) \right]. \eqno{[P]}
\end{equation*}
Here, $\sigma_n$ represents a distribution over experiments chosen by an adversarial Nature.

 It is worth noting that the DM chooses the beliefs---or precisely, a belief-formation rule---before she observes the realized signal count vector $K$.  This ex-ante perspective reflects a DM who seeks a consistent interpretation of the data-generating process across all possible sample realizations. By committing to a rule $a(K)$ before $K$ is observed, the DM ensures her inferences are robust to a fixed, albeit unknown, DGP. In contrast, an ex-post approach---where a worst-case DGP is selected separately for each realized $K$---would allow Nature to adapt the signal quality to the observed data, a scenario that would render the sample size $n$ irrelevant to the minimax regret.\footnote{If the DM were to minimize regret ex-post for each realized $K$, the worst-case DGP would be tailored to the specific signal count, resulting in a constant minimax regret that does not fall as $n$ grows. Our ex-ante formulation captures the more natural setting where the quality of the data source is uncertain but stable.}  

 In the sequel, we analyze $[P]$ by framing it as a zero-sum game between the DM and Nature.  By the well-known minimax theorem (for example, see \cite{Osborne1994} Proposition 22.2), the equilibrium of this game---called a saddle point---corresponds to the optimal solution to $[P]$.

\subsection{Payoff and Regret}

To operationalize the problem $[P]$, we must define the payoff for holding a belief.  To this end, we simply require a payoff function $S(a, \theta)$ that rationalizes the Bayesian posterior as the optimal choice. Specifically, if the DM held a hypothetical true belief $p$, her optimal action $a$ must be $a=p$.   

As characterized by \cite{savage1971elicitation}, a function $S$ satisfies this property---known as a \textbf{strictly proper scoring rule}\footnote{The context of \cite{savage1971elicitation} was to devise a payment function that elicits truthful beliefs from a agent---hence the name.}---if and only if there exists a  strictly convex function $G: [0,1] \to \mathbb{R}$ such that
\[ \mathbb{E}_p[S(a, \theta)] = G(a) + (p-a)G'(a).\]

\cite{schervish1989general} microfounds this formulation and interprets the generating function $G$.  He shows that any proper scoring rule is equivalent to the expected utility derived from a \emph{portfolio of future binary decision problems} (e.g., whether to invest, whether to treat a patient) that depend on the state $\theta$. The DM does not know which specific decision she will face, but holds a measure over the set of possible cost thresholds.  Each threshold represents a specific trade-off between the costs of Type I errors (false positives) and Type II errors (false negatives) for a given task. The ``weights'' the DM places on these future problems are captured by the curvature of $G$. Specifically, the second derivative $G''(p)$ represents the density of decision problems where the indifference threshold is near $p$. (See  \Cref{app:schervish_microfoundation} for the formal weighting representation.)

We define the \emph{regret} as the difference between the expected utility of the Oracle (who acts optimally using $q_{\text{oracle}}$) and the expected utility of the DM (who acts using belief $a$). Using Savage's representation, this difference simplifies to the \emph{Bregman Divergence} generated by $G$:
\begin{align*}
    R(q_{\text{oracle}}, a) &= \mathbb{E}_{q_{\text{oracle}}}[S(q_{\text{oracle}},\theta)] - \mathbb{E}_{q_{\text{oracle}}}[S(a, \theta)] \\
   &= G(q_{\text{oracle}})-G(a)-(q_{\text{oracle}}-a)G'(a)\\
   &=: B_G(q_{\text{oracle}} \parallel a).
\end{align*}

\subsection{Examples of Bregman Divergence}\label{subsec:example}

The DM's sensitivity to error is defined by the choice of $G$. We consider two leading examples:

\paragraph{Quadratic Loss (Mean-Squared Error).}
Let $G(p) = p^2$. The curvature is constant, $G''(p) = 2$. This implies the DM weighs errors equally across the entire probability spectrum, anticipating a uniform distribution of future decision thresholds. The resulting regret is the standard Mean-Squared Error:
\[ B_G(p \parallel a) = (p-a)^2. \]

\paragraph{Logarithmic Loss (KL Divergence).}
Let $G(p) = p \log p + (1-p) \log(1-p)$ (negative entropy). The curvature is $G''(p) = \frac{1}{p(1-p)}$, which explodes as $p$ approaches $0$ or $1$. This indicates the DM is extremely averse to errors when holding confident beliefs (near the boundaries). The resulting regret is the Kullback-Leibler divergence:
\[ B_G(p \parallel a) = D_{KL}(p \parallel a). \]
The minimax problem $[P]$ becomes:
\begin{equation*}
\min_{a(\cdot)} \max_{\sigma_n \in \Delta(\Pi)} \mathbb{E}_{\sigma_n, K} \left[ B_G\left(q_{\text{oracle}}(K; \bp) \parallel a(K)\right) \right]. \eqno{[P]}
\end{equation*}

While we allow Nature to use mixed strategies to create ambiguity, we can restrict the Decision Maker to pure strategies without loss of generality. Since the generating function $G$ is strictly convex, the resulting Bregman divergence $B_G(p \parallel a)$ is strictly convex with respect to the action $a$. Consequently, the ex-ante regret—which is a convex combination of these divergences—is strictly convex in $a$. Therefore, the optimal strategy $a^*$ is unique and deterministic.

\section{Special Case:  MSE with Binary Signals}
We begin with the special case of \emph{binary signals} and  \emph{MSE regret} (corresponding to the quadratic loss utility). This specification allows for precise closed-form characterization, which yields clear intuition and behavioral implications. We later  extend our results to a more general setting with arbitrary signal spaces and any strictly proper scoring rule.

We consider the DM holding a uniform prior $\mu = 1/2$ over the payoff-relevant state $\theta \in \{0,1\}$. She observes a sequence of $n$ i.i.d. binary signals  $s \in \mathcal{S}=\{h,l\}$, generated by a DGP indexed by its precision $\pi$. We restrict attention to symmetric DGPs where $\pi := \Pr(s=h|\theta=1) = \Pr(s=l|\theta=0)$, which is without loss of generality under the uniform prior. The DM is ambiguous about the DGP, knowing only that the precision $\pi$ lies in $\Pi = [1/2, 1]$. The lower bound $\pi \ge 1/2$ implies that signals are weakly informative and have intrinsic meaning (i.e., signal $h$ is evidence for state $1$), which rules out perverse cases where signals are systematically misleading.

Upon observing the signals, the DM chooses an action $a \in [0,1]$ to maximize the expected quadratic utility $u(a, \theta) = -(a-\theta)^2$. In our general framework (Section 2.3), this utility corresponds to the scoring rule generated by $G(p)=p^2$, resulting in the standard MSE regret. Because the signals are binary and exchangeable, the count of $h$ signals $K \equiv \sum_i \mathbf{1}_{\{s_i = h\}}$ is a sufficient statistic for the sample. We denote by $k$ the realization of the random variable $K$.
Consequently, we can represent the DM's strategy as a vector $\bm{a}_n=(a_0, \dots, a_n)$, where $a_k$ denotes the action taken (or belief formed) after observing $k$ high signals.

When Nature selects precision $\pi$ and the DM chooses $\bm a_n$, the ex-ante regret is the weighted sum of the mean squared errors across all possible realizations of $k$:
\begin{equation*}
R(\bm{a}_n,\pi)= \sum_{k=0}^n \Pr(k|\pi)(\Pr(\theta=1|\pi,k)-a_k)^2,
\end{equation*}
where $\Pr(k|\pi) = \frac{1}{2}\binom{n}{k} [\pi^k(1-\pi)^{n-k} + (1-\pi)^k\pi^{n-k}]$ is the marginal probability of observing $k$ high signals given the DGP $\pi$, and $\Pr(\theta=1|\pi,k)$ is the Oracle's posterior when $\pi$ is known.\footnote{We are abusing notation here by writing $R(\bm{a},\pi)$ as the ex-ante regret when Nature chooses $\pi$ and DM plays $\bm a$. This is different from the $R\left(q_{\text{oracle}}(K; \bp), a(K)\right)$ in the last section, which is the interim regret that depends on the realization of $K$.}

Since Nature may randomize over DGPs via a mixed strategy $\sigma_n$, the minimax problem is:
\begin{equation*}
\min_{\bm a_n} \max_{\sigma_n \in \Delta(\Pi)} \mathbb{E}_{\sigma_n} \left[\sum_{k=0}^n \Pr(k|\pi)(\Pr(\theta=1|\pi,k)-a_k^2 \right].
\end{equation*}

\subsection{Learning from Finite Sample}
\subsubsection{Learning from a Single Signal Draw}

We start by analyzing the case where the DM receives only one signal draw. 
\begin{proposition}\label{prop1}
When $n=1$, in the equilibrium Nature chooses $\pi$ that randomizes between $\frac{1}{2}$ and $1$ with equal probability. The DM chooses $(a_0,a_1)=(\frac{1}{4},\frac{3}{4})$.
\end{proposition}

\begin{proof}
See appendix \Cref{pf:prop1}.
\end{proof}

\begin{figure}[htbp]
    \centering
    \includegraphics[width=0.6\linewidth]{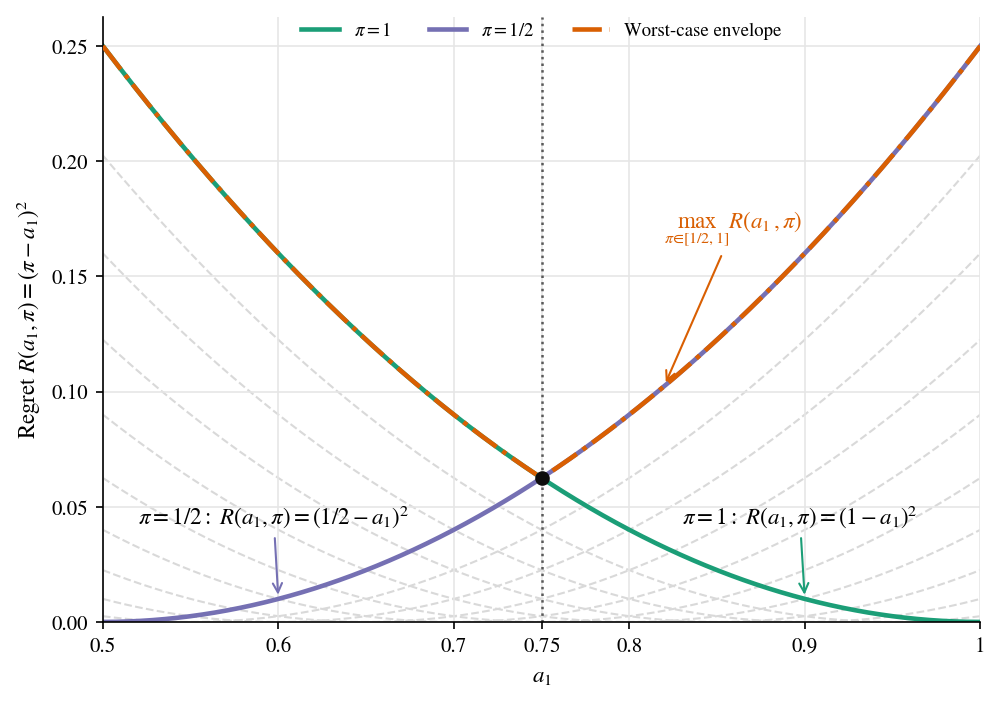}
    \caption{Illustration for $n=1$ case} 
    \label{fig1}
\end{figure}

\Cref{fig1} illustrates this solution. The horizontal axis is $a_1$ (by symmetry, $a_0 = 1-a_1$), and the vertical axis is regret.  For each signal precision $\pi \in [\tfrac12,1]$, the regret function $(\pi-a_1)^2$ traces out a curve over $a_1$; each curve corresponds to a different value of $\pi$. The pointwise upper envelope of these curves is attained by the two extreme cases, $(1-a_1)^2$ and $(\tfrac12-a_1)^2$, corresponding to $\pi=1$ and $\pi=\tfrac12$, respectively. For any $\pi \in (\tfrac12,1)$, the curve $(\pi-a_1)^2$ lies strictly below this upper envelope for all $a_1$, so every such $\pi$ is strictly dominated by either $\pi=\tfrac12$ or $\pi=1$.  

Clearly, Nature can't choose a pure strategy in equilibrium.  If Nature chooses a pure strategy, it must be either $\pi=1/2$ or $\pi=1$, given the above observation.  Suppose Nature chooses a fully revealing DGP ($\pi=1$).  The DM  then responds with a belief $a_1=1$, eliminating regret. However, Nature then deviates to a fully uninformative DGP ($\pi=1/2$), again maximizing the regret; in response, the DM switches to a belief $a_1=1/2$, which in turn triggers Nature to deviate to $\pi=1$.  In equilibrium, Nature randomizes over two DGP's with equal probability, and the DM chooses a Bayesian posterior $a_1=\tfrac34$ against these adversarial DGPs. Consistent with the minimax theorem, this saddle point coincides with a minimax solution:  $a_1=\tfrac34$ minimizes the upper envelope depicted in  \Cref{fig1}.

\subsubsection{Learning from Multiple Signal Draws}

We now extend the analysis to the case where the DM can observe $n>1$ signal draws. As the sample size increases, Nature's optimal strategy changes: rather than randomizing between the extremes ($1/2$ and $1$), Nature begins to place probability mass on a specific ``interior'' informative signal quality.
\begin{theorem}\label{thm1}
We have following statement:
\begin{enumerate}
\item 
When $n=2$, in the equilibrium Nature chooses $\sigma$ that randomizes between $\frac{1}{2}$ and $1$ with probability $2-\sqrt{2}$ and $\sqrt{2}-1$, respectively. 
\item 
When $n=3$, in the equilibrium Nature chooses $\sigma$ that randomizes between $\frac{1}{2}$ and $\pi^*$ for some $\pi^*\in(\frac{1}{2},1)$. 
\end{enumerate}
\end{theorem}

\begin{proof}
    See appendix \Cref{pf:thm1}.
\end{proof}

We show that when \(n=2\), Nature continues to randomize between the fully revealing DGP (\(\pi=1\)) and the uninformative DGP (\(\pi=1/2\)), but assigns relatively more weight to the uninformative DGP. When \(n=3\), Nature instead places positive probability on an interior informative precision \(\pi^*\in(1/2,1)\) together with the uninformative DGP.

For \(n=2\), Nature primarily distorts information by putting more weight on the uninformative DGP; while for \(n=3\) Nature distorts \emph{informativeness} itself by introducing an interior precision \(\pi^*\). This raises a natural question: as the sample size grows, how does Nature manipulate the information environment? There are three conceptually distinct channels:
\begin{enumerate}
\item Place more mass on the uninformative endpoint \(\pi=1/2\), making data frequently uninformative;
\item Reduce the informativeness of the informative component by moving its support inward; or
\item Do both simultaneously---shifting more weight to \(\pi=1/2\) while also moving the interior support point toward \(1/2\).
\end{enumerate}
We will show that channel (2) is the dominant force as \(n\) increases. Intuitively, larger samples expand the set of sufficient statistics and create a richer set of contingencies, allowing the DM to hedge her response---for example, reacting strongly to extreme counts while tempering her response at interior counts. If Nature were to mix only between \(\pi=1/2\) and \(\pi=1\), the \(\pi=1\) component would become too easy to exploit: extreme realizations would be perfectly revealing, enabling the DM to achieve very low loss in those cases. Introducing an interior precision blunts this advantage by making even extreme outcomes less decisive.

A closed-form characterization for \(n>3\) is computationally intractable using our approach. Instead, we provide numerical evidence for the conjectured equilibrium by imposing the necessary equilibrium conditions implied by its structure. We then use an asymptotic analysis to formalize and prove the underlying mechanism.

Specifically, we conjecture that Nature randomizes between the uninformative experiment $\pi=\tfrac12$ and an informative experiment with precision $\pi_n^*$, assigning probability $w \in [0,1]$ to the latter. Under this conjecture, an equilibrium is pinned down by three conditions. (1) Given Nature’s postulated mixed strategy, the DM’s action must be a best response; equivalently, it satisfies the corresponding first-order condition (i.e., Bayes’ rule evaluated at the worst-case prior). (2) Nature must be indifferent between the two experiments $\pi=\half$ and $\pi=\pi_n^*$ when facing the DM’s best response for it to randomize. (3) The informative precision $\pi_n^*$ must locally maximize Nature’s objective (the regret) among nearby feasible precisions.

\begin{align*}
   &  a_k^*  =\frac{2^{n}w{\pi_n^*}^k(1-\pi_n^*)^{n-k}+1-w}{2^{n}w({\pi_n^*}^k(1-\pi_n^*)^{n-k}+(1-\pi_n^*)^k{\pi_n^*}^{n-k})+2(1-w)}. \tag{FOC}\label{eq:foc}\\
   &  R(\bm{a}^*,\pi_n^*) =R(\bm{a}^*,\frac{1}{2}) \Leftrightarrow  \sum_{k=0}^n\Pr(k|\pi_n^*)(\Pr(\t=1|\pi_n^*,k)-a_k^*)^2 =\frac{1}{2^{n}}\sum_{k=0}^n{n \choose k}(\frac{1}{2}-a_k^*)^2.\tag{Indif}\label{eq:indifference}\\
&\frac{\partial}{\partial \pi} R(\bm{a}^*, \pi) \bigg|_{\pi=\pi_n^*} = \sum_{k=0}^{n} \Bigg[ \frac{\partial \Pr(k|\pi)}{\partial \pi} (\Pr(\theta=1|\pi,k) - a_k^*)^2 \\
&\qquad\qquad\qquad + 2 \Pr(k|\pi) (\Pr(\theta=1|\pi,k) - a_k^*) \frac{\partial \Pr(\theta=1|\pi,k)}{\partial \pi}\Bigg]_{\pi=\pi_n^*} = 0.\tag{Local Opt}\label{eq:localoptimality}
\end{align*}

For different $n$, we can solve the solution $((a_k)_{0\le k\le n}, \pi,w)$ to the system of equations given by these three conditions by simulation. 

\begin{figure}[htbp]
    \centering
    \includegraphics[width=1\linewidth]{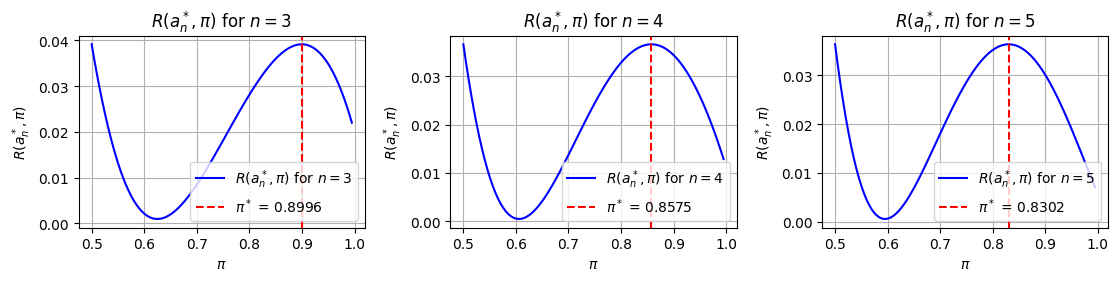}
    \caption{Simulation for $3\le n\le 5$}
    \label{figc}
\end{figure}

  \Cref{figc} displays the results for $3 \le n \le 5$. The horizontal axis represents $\pi$, and the vertical axis represents regret, evaluated at the DM’s equilibrium strategy $\bm a^*_n$. In each panel, Nature is indifferent between two optimal experiments: an uninformative one ($\pi=1/2$) and an interior informative one ($\pi=\pi_n^*\in(1/2,1)$). A clear pattern emerges: as $n$ increases, the optimal interior informative precision $\pi_n^*$ for Nature decreases toward $1/2$. Intuitively, Nature strategically degrades the signal quality to prevent the DM from learning the state too easily. In the next section, we analyze the limit of this game to determine the precise rate at which this degradation occurs. 

 \begin{figure}[h]
    \centering
    \includegraphics[width=1.0\linewidth]{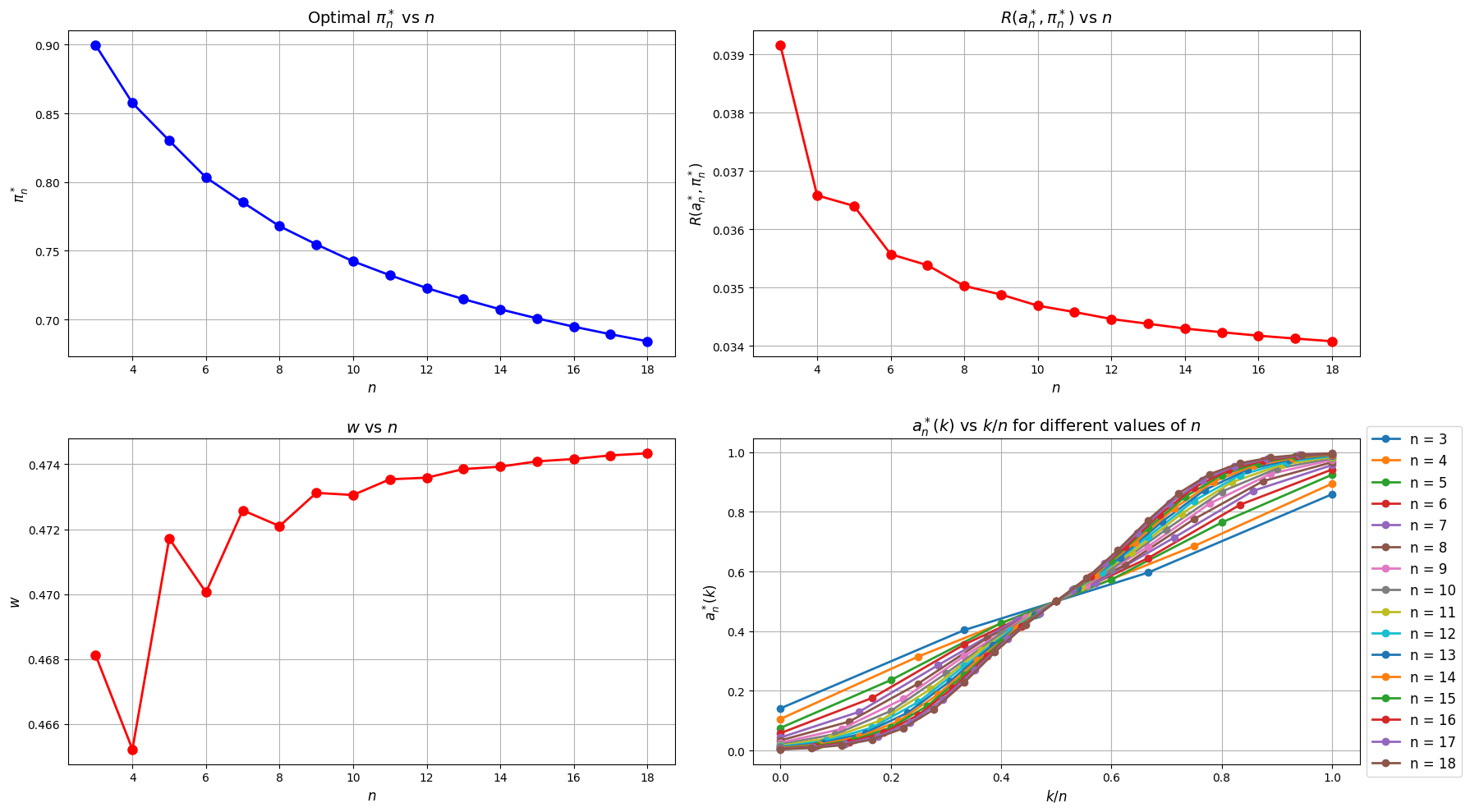}
    \caption{Trend of $\pi_n^*, R(\bm a^*_n, \pi_n^*), w, (a^*_n(k))_{0\le k\le n}$ for different $n$}
    \label{figcc}
\end{figure}   

 The upper left panel of  \Cref{figcc} shows a more direct pattern that the optimal informative experiment $\pi_n^*$ decreases with $n$ for $3\le n\le 18$. The upper right panel shows that the equilibrium regret $R(\bm a_n^*, \pi_n^*)$ also decreases with $n$. This indicates that the data size helps reduce the worst case regret. The lower left panel shows the weight Nature assigns to the optimal informative experiment $\pi_n^*$ fluctuates for small $n$ but exhibits an overall upward trend as $n$ increases. Finally, the lower right panel plots each equilibrium action rule (DM's posterior) $a^*_n(k)$ for signal fraction $k/n$ across different $n$'s, represented by different curves.

\subsection{Limit Learning from Large Sample}

In the small $n$ analysis and the simulations, we observe that as the sample size $n$ grows, Nature's optimal signal precision $\pi_n^*$ shifts toward the uninformative boundary $1/2$. In this section, we formally analyze the asymptotic behavior of the MSE game.

In principle, one could study the equilibrium of the finite game directly for each \(n\). In practice, however, this approach quickly becomes infeasible: the closed-form characterizations become computationally difficult once \(n\) grows large. More importantly, the simulated equilibrium suggests a specific pattern: Nature’s optimal precision \(\pi_n^*\) drifts toward the uninformative boundary \(1/2\), so the economically relevant question is how close to \(1/2\) Nature chooses \(\pi\) as \(n\) grows.

The key insight, to be established later, is that adversarial Nature adopts equilibrium strategy in which the informative DGP moves closely to the uninformative at $1/\sqrt{n}$ rate. With this ``scaling'' of  $1/\sqrt{n}$,\footnote{We do emphasize the ``scaling'' here does not mean normalization, but a local reparameterization that zooms in on a shrinking $1/\sqrt{n}$-neighborhood of $\pi=1/2$.} the finite-sample log-likelihood ratio generated by \(K\) admits a Gaussian approximation, so the binomial experiment converges to a simple normal-shift experiment. The resulting \emph{limit game} captures the leading-order incentives of both players in the large-\(n\) game and delivers a tractable characterization of Nature’s equilibrium play and the DM’s response in the limit.  

We first define a \emph{Finite Game} and a continuous \emph{Limit Game}. We then characterize the unique equilibrium of the Limit Game, which describes the robust learner's behavior in the presence of massive data. Finally, we establish a formal convergence result, linking the finite-sample game to its limit game counterpart, showing that equilibria of the finite games are well-approximated by the equilibrium of the limit game for large \(n\).

\subsubsection{The Finite and Limit Games}

\paragraph{\textcolor{deepblue}{{The  Finite  Game  $\Gamma_n$.}}}
For each $n$, in the finite game Nature's strategy is to choose precision $\pi \in [1/2, 1]$. We reparameterize this using the local parameter $c_n = \sqrt{n}(\pi - 1/2)$, which lies in the interval $[0, \sqrt{n}/2]$. Thus, Nature's strategy is $\gamma_n\in \Sigma_n\equiv\Delta[0, \sqrt{n}/2]$. The DM chooses an action profile (her posterior beliefs) $\bm{a}_n = (a_0, \dots, a_n)$, which is a vector in $A_n\equiv [0,1]^{n+1}$. The payoff of Nature is the expected ex-ante MSE Regret 
\[
R_n(\bm a_n, \gamma_n)\equiv\E_{c_n\sim\gamma_n}\left[\sum_{k=0}^n\Pr(k|c_n)(\Pr(\theta=1|k,c_n)-a_k)^2\right],
\]
while the DM gets the negative of this payoff as in a zero-sum game.

And the finite game can be formally defined by the tuple $\Gamma_n = \langle \{DM, \text{Nature}\}, \{A_n, \Sigma_n\}, R_n \rangle$.

\paragraph{\textcolor{deepblue}{The Limit Game $\tilde{\Gamma}$.}}
We define the Limit Game $\tilde{\Gamma}=\langle \{DM, \text{Nature}\},\{A,\Sigma\},\tilde{R}\rangle$ directly in terms of the Gaussian limit experiment. The Nature's strategy is a probability measure $\gamma \in \Sigma = \Delta([0, \infty))$ over a local precision parameter $c \ge 0$. This choice of precision affects the signal structure the DM receives. The DM observes a signal realization $z \in \mathbb{R}$, whose mean depends on the state $\theta$ and precision $c$:
    \[ Z \mid \theta=1, c \sim \mathcal{N}(2c, 1) \quad \text{and} \quad Z \mid \theta=0, c \sim \mathcal{N}(-2c, 1). \]
In response, the DM chooses action $\bm a: \mathbb{R} \to [0,1]$, which is a monotone function mapping the signal realization $z$ to a posterior belief. Nature gets the limit regret 
\begin{align*}
    \tilde{R}(\bm{a},\gamma)&=\frac{1}{2} \E_{c\sim\gamma} \int_{-\infty}^{\infty} \phi(z+2c)(\Pr(\theta=1|z,c)-\bm a(z))^2 dz\\
&+\frac{1}{2} \E_{c\sim\gamma} \int_{-\infty}^{\infty} \phi(z-2c)(\Pr(\theta=1|z,c)-\bm a(z))^2 dz
\end{align*}
as payoff, while the DM gets the negative of it.

\subsubsection{Characterizing the Limit Equilibrium}

Next, we derive explicit expressions for the beliefs and strategies that solve the Limit Game.

\begin{theorem}\label{thm:limit_eq}
The Limit Game $\tilde{\Gamma}$ has a unique saddle point $((c^*, w^*), \bm a^*)$ characterized as follows:
\begin{enumerate}
    \item \textbf{Nature's Strategy:} Nature randomizes between exactly two points: $c=0$ (uninformative) and a unique informative precision $c^* > 0$, with probabilities $1-w^*$ and $w^*$ respectively.
    
    \item \textbf{DM's Strategy:} The DM employs the Bayesian posterior $\bm a^*(z)$ generated by Nature's worst-case strategy, explicitly:
    \[ \bm a^*(z) = \frac{1-w^*+w^* e^{2c^*z-2{c^*}^2}}{2(1-w^*)+w^*(e^{2c^*z-2{c^*}^2}+e^{-2c^*z-2{c^*}^2})}. \]
\end{enumerate}
\end{theorem}

The constants can be calculated by solving the indifference condition ($\tilde{R}(\bm a^*,c^*) = \tilde{R}(\bm a^*,0)$) and the optimality condition ($\frac{\partial}{\partial c} \tilde{R}(\bm a^*,c^*) = 0$). Numerically, we find $c^* \approx 0.799$ and $w^* \approx 0.476$.

 This proposition shows Nature’s worst-case strategy collapses to a \emph{two-point} mixture, even though it is allowed to randomize over a continuum of precision levels. The two-point mixing \(\{0,c^*\}\) reflects how Nature maximizes confusion under  massive-data. Assigning all weight to very large \(c\) is counterproductive for Nature because it separates the two Gaussian signal distributions, making the state almost surely identifiable and driving regret to zero. Assigning all weight to \(c=0\) makes signals pure noise, in which case both the oracle and the DM simply stays near the prior and Nature’s ability to induce regrets is limited.  

 In equilibrium, Nature combines these forces by placing some mass on \(c=0\) (so any realized \(z\) may be noise) and the remaining mass on a specific interior \(c^*>0\) (so the same \(z\) may also be informative). This structure makes evidence appear potentially decisive while preserving doubt about its reliability. The equilibrium pair \((w^*,c^*)\) balances two levers: \(w^*\) controls how often informative evidence arrives, while \(c^*\) controls how strong it is when it arrives. The saddle point selects \((w^*,c^*)\) so that signals are strong enough to pull the DM away from \(1/2\), but not so strong that the DM can exploit them safely.

\paragraph{Proof Sketch.}
Let $\bm a^*(c^*,w^*)$ be DM's best response to Nature's strategy randomizing between two points: $c=0$ and $c^* > 0$, with probabilities $1-w^*$ and $w^*$ respectively in the limit game.  

The necessary condition for Nature's strategy to randomize between $0$ and $c^*$, with probabilities $1-w^*$ and $w^*$ is that the pair $(c^*,w^*)$ solves the following conditions:
\begin{enumerate}
\item 
Nature is indifferent between $c=c^*$ and $c=0$.
\begin{equation}\label{indiff}
\tilde{R}(\bm a^*_n(c^*,w^*), c^*)-\tilde{R}(\bm a^*_n(c^*,w^*), 0)=0.   
\end{equation}
\item 
$c^*$ satisfies the first-order condition 
for nature's optimization problem.
\begin{equation}\label{optim}
  \frac{\partial}{\partial c}\tilde{R}(\bm a^*_n(c^*,w^*), c^*)=0.
\end{equation}
\end{enumerate}
We prove that there exists a unique $(c^*,w^*)$ that solves this system of equations. Moreover, this solution, together with DM's best response, forms a saddle point. To prove $((c^*, w^*), \bm a^*)$ is indeed a saddle point, we also need to show $(c^*,w^*)$ is globally optimal. We show that by characterizing the regret function only has two peaks, combining with \Cref{indiff} and boundary behavior of the regret function indicate the global optimality. The complete proof is provided in appendix \Cref{pf:thm:limit_eq}.

\subsubsection{Convergence of Finite Games}
We now link the sequence of finite games $\Gamma_n$ to the limit game $\tilde{\Gamma}$. We prove the sequence of Nash equilibria of the finite games converges to the \emph{unique} Nash equilibrium of the limit game.

Before presenting the result, note that these strategies are not defined on a common space. To discuss rigorously the convergence of $\Gamma_n$ to $\tilde{\Gamma}$, we first embed the finite- and limit-game strategy spaces into a single compact metric space.

\bigskip
\paragraph{Embedding Strategies.} For Nature's strategy, we embed the finite game's strategy space $\Sigma_n$ into the limit game's strategy space $\Sigma$. In the finite game, Nature chooses a measure $\gamma_n\in \Sigma_n\equiv\Delta[0, \sqrt{n}/2]$. As $n \to \infty$, the domain expands to $[0, \infty)$. We view these as measures on the common compact domain $[0, C]$.\footnote{\label{fn:precision}We can restrict $c$ to a sufficiently large compact set $[0, C]$ without loss of generality, as infinite precision yields zero regret and is strictly dominated for Nature.} Let $\Sigma = \Delta([0, C])$ be the space of probability measures, which is compact under the Prokhorov metric (weak convergence).


For DM's strategy, define a strategy space $\mathcal{A}$ as the set of \emph{monotone non-decreasing functions} $f: [0,1] \to [0,1]$. We embed the finite game's strategy space $A_n$ and the limit game's strategy space $A$ into this strategy space $\mathcal{A}$. For \emph{limit strategies}, given $\bm a:\mathbb R\to[0,1]$, define its embedding $\tilde{\bm a}\in\mathcal A$ by $\tilde{\bm a}(y)=\bm a(\Phi^{-1}(y))$ for $y\in[0,1]$, where $\Phi$ is the standard normal CDF. For \emph{finite strategies}, given $\bm a_n=(a_0,\ldots,a_n)\in A_n$, map each signal count $k$ to its asymptotic quantile under the uninformative experiment $c=0$ (equivalently $\pi=1/2$) via $y_{n,k}=\Phi\left(\frac{k - n/2}{\sqrt{n}/2}\right)$, and define $\tilde{\bm a}_n\in\mathcal A$ as the piecewise-linear interpolation of the points $\{(y_{n,k},a_k)\}_{k=0}^n$.
By Helly's Selection Theorem, $\mathcal{A}$ is compact under pointwise convergence. 

\bigskip

Given this embedding, the next theorem shows the Nash equilibria in the finite games converge to the unique Nash equilibrium in the limit game as $n\to \infty$.

\begin{theorem}[Convergence of Equilibria]\label{thm2}
Let $(\tilde{\bm{a}}_n^*, \gamma_n^*)$ be a sequence of Nash equilibria for the finite games $\Gamma_n$. Let $(\tilde{\bm{a}}^*, \gamma^*)$ be the unique Nash Equilibrium of the limit game $\tilde{\Gamma}$.
Then, as $n \to \infty$:
\begin{enumerate}
    \item \textbf{Nature:} $\gamma_n^* \xrightarrow{d} \gamma^*$ (weak convergence of measures).
    \item \textbf{DM:} $\tilde{\bm a}_n^*(z) \to \tilde{\bm{a}}^*(z)$ for all continuity points $z \in [0,1]$ (pointwise convergence).
\end{enumerate}
\end{theorem}

\begin{proof}
    See appendix \Cref{pf:thm2}.
\end{proof}

  \Cref{thm2} provides a formal justification for using the limit game as an asymptotic description of the finite-sample problem. This result also answers the earlier question. Since \(\gamma_n^* \xrightarrow{d} \gamma^*\) and \(\gamma^*\) assigns constant positive probabilities \((1-w^*,\,w^*)\) to exactly two points \(\{0,c^*\}\), Nature’s large-\(n\) manipulation does not operate primarily by shifting more and more mass to the uninformative endpoint. Instead, Nature sustains regret mainly by selecting a \emph{specific interior informative strength} for the informative component (the analogue of \(c^*\)) and mixing it with an uninformative component.

\subsubsection{Implications: Incomplete Learning Guarantees}

The existence of this non-trivial equilibrium has profound implications for long-run learning. In particular, the equilibrium regret is bounded away from zero, so learning against the worst case DGP is incomplete. Note that this regret is the equilibrium guarantee against the worst-case DGP.  \Cref{sec:price_of_robustness} by contrast considers a fixed informative DGP that is not chosen by an adversarial Nature, in which case the learning implications differ. Another implication is that the robust belief does not converge to the true state even with infinite data.  

\begin{corollary}[Failure of Law of Large Numbers against adversarial nature]\label{co2}
With adversarial Nature, in the limit as $n\to\infty$:
\begin{enumerate}
    \item \textbf{Non-Vanishing Regret:} The ex-ante regret remains strictly positive.
    \item \textbf{Failure of Asymptotic Learning:} The DM's robust posterior $\bm a^*(Z)$ does not converge to the true state $\theta$. Instead, it converges to a non-degenerate random variable strictly bounded away from 0 and 1.
\end{enumerate}
\end{corollary}

\begin{proof}
    See appendix \Cref{pf:co2}.
\end{proof}

Nature’s equilibrium strategy in the large-$n$ limit is to randomize its precision between a completely uninformative signal and one that provides just enough information to keep the DM's regret from vanishing. Specifically, Nature chooses a process whose precision converges to the uninformative level at the rate $c^*/\sqrt{n}$. This strategic blurring of the signal ensures that even though the sample size grows, the total information available to the DM remains bounded, preventing full learning and maintaining a strictly positive level of ex-ante regret.

While the $1/\sqrt{n}$ scaling is reminiscent of the ``local alternatives'' framework in asymptotic statistics (e.g., \cite{hirano2009asymptotics} for example), the underlying mechanism here is fundamentally different. In that tradition, signal precision is typically assumed to be fixed and informative, which would—by the Law of Large Numbers—eventually resolve all uncertainty regarding fixed, ``global'' alternatives. To study non-trivial decision trade-offs, researchers must therefore zoom in on states that are microscopic distances apart ($1/\sqrt{n}$). In our framework, the states $\theta=0$ and $\theta=1$ are global and fixed. The non-degenerate limit experiment arises not from the proximity of the states, but from the strategic ambiguity of the signal quality. In this sense, we provide a {\it decision-theoretic dual} to the local alternatives method: rather than scaling the parameter space to match the data, we show that an adversarial Nature scales the precision space to match the sample size.

\subsubsection{Proof Sketch for \Cref{thm2}}

The proof of \Cref{thm2} (found in \Cref{pf:thm2}) establishes that any sequence of Nash equilibria $(\bm{\tilde{a}}_{n}^{*}, \gamma_{n}^{*})$ of the finite games $\Gamma_{n}$ converges to the unique Nash equilibrium $(\bm{\tilde{a}}^{*}, \gamma^{*})$ of the limit game $\tilde{\Gamma}$. The argument proceeds in four concrete steps:

\medskip
\noindent \textbf{Step 1: Embedding and Existence.} Because the finite and limit games are defined on different strategy spaces, we first embed them into a common space. Nature's strategies $\gamma_n$ are viewed as measures on a common compact domain $[0, C]$, where $C$ is large enough that higher precisions are strictly dominated. For the DM, we define the strategy space $\mathcal{A}$ as the set of monotone non-decreasing functions mapping $[0,1]$ to $[0,1]$. By Helly’s Selection Theorem, $\mathcal{A}$ is compact under pointwise convergence. Existence of a saddle point for each $n$ then follows from Sion's Minimax Theorem, as the strategy spaces are compact and convex, and the regret function $R_n(\bm{\tilde{a}}, \gamma)$ is linear in $\gamma$ and strictly convex in $\bm{\tilde{a}}$.

\medskip
\noindent \textbf{Step 2: Pointwise Convergence of the Regret Function.} We prove that for any fixed strategy pair $(\bm{\tilde{a}}, \gamma) \in \mathcal{A} \times \Sigma$, the finite-sample regret converges to the limit regret:
\begin{equation*}
    R_n(\bm{\tilde{a}}, \gamma) \longrightarrow \tilde{R}(\bm{\tilde{a}}, \gamma) \text{ as } n \to \infty.
\end{equation*}
As detailed in \Cref{prop:limit regret}, this utilizes a Local Central Limit Theorem to show that binomial counts converge to Gaussian densities and that the discrete log-likelihood ratios converge to the linear shift $4cz$.

\medskip
\noindent \textbf{Step 3: Two-Stage Uniform Convergence.} To link the players' optimization problems across $n$'s, we define Nature's value functions in the finite and limit games:
\begin{equation*}
    V_n(\gamma) := \inf_{\bm{\tilde{a}}_n \in \mathcal{A}} R_n(\bm{\tilde{a}}_n, \gamma), \quad V_{\infty}(\gamma) := \inf_{\bm{\tilde{a}} \in \mathcal{A}} \tilde{R}(\bm{\tilde{a}}, \gamma).
\end{equation*}
We establish convergence in two stages:
\begin{enumerate}
    \item \textbf{Uniformity over $\mathcal{A}$:} For any sequence $\gamma_n \to \gamma$, we show that $R_n(\cdot, \gamma_n)$ converges uniformly to $\tilde{R}(\cdot, \gamma)$ on $\mathcal{A}$. This follows from the fact that on a compact convex set, pointwise convergence of convex functions to a continuous limit implies uniform convergence. This ensures that the infima converge: $V_n(\gamma_n) \to V_{\infty}(\gamma)$.
    \item \textbf{Uniformity over $\Sigma$:} We then show that the sequence of value functions converges uniformly across all of Nature's strategies:
    \begin{equation*}
        \sup_{\gamma \in \Sigma} |V_n(\gamma) - V_{\infty}(\gamma)| \longrightarrow 0.
    \end{equation*}
    This follows because each $V_n$ is a concave function on the compact set $\Sigma$, and their pointwise convergence to the continuous $V_{\infty}$ implies uniform convergence.
\end{enumerate}

\medskip
\noindent \textbf{Step 4: Strategy Convergence via Berge's Maximum Theorem.} Finally, we apply Berge's Maximum Theorem to Nature's maximization problem, $\max_{\gamma \in \Sigma} V_n(\gamma)$. Since $V_n \to V_{\infty}$ uniformly and the limit maximizer $\gamma^*$ is unique (as proven in \Cref{thm:limit_eq}), any sequence of finite maximizers $\gamma_n^*$ must converge weakly to $\gamma^*$. Because the DM's best response to any $\gamma$ is unique, the uniform convergence of the regret functions ensures that the associated minimizers also converge pointwise: $\bm{\tilde{a}}_n^* \to \bm{\tilde{a}}^*$.

\subsection{Robust Law of Large Numbers and the Price of Robustness} \label{sec:price_of_robustness}

In the previous sections, we derived the DM's robust strategy $\bm a_n^*$ when Nature is adversarial. Here, we analyze the performance of this strategy when the data is actually generated by a fixed, non-adversarial DGP.

Suppose the true DGP is fixed with precision $\pi_{\text{true}} \in (1/2, 1]$. The DM, unaware of this, employs the minimax regret strategy $\bm a_n^*$ designed for the worst-case scenario. We show the DM's actions are consistent and her error converges to zero: the DM eventually learns the true state with enough data. However, her error vanishes at a sub-exponential rate, which is strictly slower than the exponential rate achieved by an ``Oracle'' who knows the true DGP. This difference quantifies the asymptotic cost of robustness.

We define three measures.   The Oracle's expected Mean Squared Loss conditional on $\t$ is\footnote{Since the prior $\mu=\frac{1}{2}$ is fixed, the unconditional Mean Squared Loss is a fixed convex combination of the two conditional Mean Squared Losses, so conditioning on $\theta$ is without loss for the asymptotic rate (and by symmetry the two conditional errors are equal).}
\begin{equation*}
L_n^{\text{oracle}}\equiv\E_{K|\pitrue, \tstate}\left[ (\Prb(\tstate=1|\pitrue, K) - \tstate)^2 \right].
\end{equation*}
The DM's Mean Squared Loss (under the robust belief) conditional on $\t$ is
$$
L_n \equiv \E_{K|\pitrue, \tstate}\left[ (\bm a_n^*(K) - \tstate)^2 \right].
$$
The DM's expected regret under this misspecified model, is defined as
\begin{equation*}
    R_n^{\text{mis}} \equiv \E_{K|\pitrue}\left[ \left( \Prb(\tstate=1|\pitrue, K) - \bm a_n^*(K) \right)^2 \right].
\end{equation*}

We use standard asymptotic notation: for nonnegative sequences $f(n)$ and $g(n)$, $f(n)=O(g(n))$ means there exist constants $M>0$ and $n_0$ such that $f(n)\le M g(n)$ for all $n\ge n_0$; $f(n)=\Theta(g(n))$ means there exist constants $0<c\le M$ and $n_0$ such that $c g(n)\le f(n)\le M g(n)$ for all $n\ge n_0$.

\begin{theorem}\label{thm:mis}
$L_n^{\text{oracle}}$ converges to zero at an exponential rate, while $L_n$ and $R_n^{\text{mis}}$ converge to zero at a sub-exponential rate. Specifically, there exists a constant $\Xi:=4c^*(2\pi_{\text{true}}-1)$ such that
\begin{equation*}
L_n^{\text{oracle}} = O({e^{-nD_{KL}(\frac{1}{2} || \pi_\text{true})}})   
\end{equation*}
$$
L_n = \Theta \left(e^{-\Xi\sqrt{n}}\right).
$$
$$
R_n^{\text{mis}} = \Theta \left(e^{-\Xi\sqrt{n}}\right).
$$
\end{theorem}

\begin{proof}
    See   \Cref{pf:thm:mis}.
\end{proof}

The result quantifies the ``price'' of robustness. It shows that under a fixed informative DGP $\pi_{\text{true}}$, the Oracle's mean squared loss converges to zero as $n\to\infty$, at a rate at least as fast as the exponential rate of $e^{-nD_{KL}(\frac{1}{2} || \pi_\text{true})}$. The DM's mean-squared loss and expected regret also vanish, but only at a sub-exponential rate of $e^{-\Xi\sqrt{n}}$.  This means that, despite the misspecification, the DM attains full learning with enough data---a robust version of Law of Large Numbers---albeit at a slow rate.  
 
 Note that the DM's rate of the constant $\Xi=4c^*(2\pi_{\text{true}}-1)$ depends both on the constant $c^*$---the interior support point in worst-case precision in the limit---and on the true precision $\pi_{\text{true}}$. In particular, a more informative true DGP implies a larger $\Xi$ and hence faster learning for the DM. 
 
 Because the robust strategy guards against local alternatives (i.e., $\pi \approx \tfrac12 + c/\sqrt{n}$), it is overly conservative: it downweights the signal when the truth is far from uninformative ($\pi_{\text{true}} \gg \tfrac12$) and the signal is in fact much stronger. Although the DM eventually learns the truth—consistent with the law of large numbers—learning is substantially slower than under commitment to a standard Bayesian model.
\smallskip

\subsubsection{Proof Sketch}  To prove this theorem, first we show the DM's ex-ante mean squared loss $L_n$ converges to zero at a sub-exponential rate. Then we can rewrite the regret term:
$$
\E_{K|\pitrue, \tstate=1}\left[ \left( q_{n,\text{oracle}}^{\text{true}}(K) - \bm a_n^*(K) \right)^2 \right] = \E_{k|\pitrue, \tstate=1}\left[ (1 - \bm a_n^*(K))^2 \left( 1 - \frac{1-q_{n,\text{oracle}}^{\text{true}}(K)}{1-\bm a_n^*(K)} \right)^2 \right].
$$
where $q_{n,\text{oracle}}^{\text{true}}(K) = \Prb(\tstate=1|\pitrue, K)$ is the Oracle's posterior.

This decomposition reveals that the regret depends on the ratio of the Oracle's error to the DM's error, and the DM's ex-ante mean squared loss $L_n$. We show the ratio $\frac{1-q_{n,\text{oracle}}^{\text{true}}(K)}{1-\bm a_n^*(K)}$ converges to zero, meaning that whenever the DM still makes a nontrivial error, the oracle's error is already negligible. Thus convergence rate of the misspecified regret is governed by the DM's mean squared loss, i.e., $R_n^{\text{mis}}$ decays at the same rate as $L_n$.

To prove the Oracle's convergence rate, we split the error sum into two parts:
\[
L_n^{\text{oracle}} = \underbrace{\sum_{k=0}^{\lfloor n/2 \rfloor} P(k)(1-q_{n,\text{oracle}}^{\text{true}}(k))^2}_{\text{Part A: Misleading Samples}} + \underbrace{\sum_{k=\lfloor n/2 \rfloor+1}^{n} P(k)(1-q_{n,\text{oracle}}^{\text{true}}(k))^2}_{\text{Part B: Correctly Classified Samples}}
\]

We use large deviation theory to bound the error from misleading samples, and use Robbins' sharpened Stirling bounds to bound the error from correctly classified samples.

Intuitively, the misspecified regret is the squared difference between these two learners' beliefs, and the mean squared loss is the squared difference between the DM and the absolute truth. The Oracle learns the truth so much faster than the DM that, from the DM's slow perspective, the Oracle and the truth become indistinguishable.

\subsection{Over-Inference and Under-Inference}

Our results also have behavioral implications for when a robust DM over- or under-infers from the data. We find that the sample size $n$ is the key determinant: the misspecified DM over-infers only when the dataset is small, whereas with large samples the DM systematically under-infers.

\begin{proposition}\label{prop:inference}
We have the following statements:
\begin{enumerate}
    \item  There exist $n$ and true precision levels $\pi_{\text{true}}$ for which the DM over-infers. That is, 
    \[|\bm a_n^*(k) - \half| > |q_{n,\text{oracle}}^{\text{true}}(k) - \half|.\]
    \item  As $n \to \infty$, the DM under-infers with probability $1$. Namely
    \[ \lim_{n\to\infty} \Pr\left( |\bm a_n^*(K)-\half| < |q_{n,\text{oracle}}^{\text{true}}(K)-\half| \right) = 1. \]
\end{enumerate}
\end{proposition}

\begin{proof}
    See appendix \Cref{pf:prop:inference}.
\end{proof}

When $n$ is small, the DM has only a few observations to rely on. In that case, treating an informative signal as uninformative effectively discards the main source of information, which is particularly costly. As a result, robustness concerns about the possibility of an uninformative signal can make the DM react too strongly to the data, generating over-inference. 

As $n$ grows, however, the DM can extract reliable information via the law of large numbers. In large-sample setting, it becomes more costly to mistake uninformative noise for an informative signal. The same robust strategy then makes the DM overly conservative, leading to systematic under-inference relative to the Oracle.

\section{General Case: Bregman Divergence and Finite Signals}

The analysis in the preceding section assumes binary signals and a quadratic loss function (MSE). A natural question is whether the $1/\sqrt{n}$ rate of ambiguity—and the resulting incomplete learning—is an artifact of this specific specification. 

In this section, we prove that this rate is a fundamental property of robust learning that holds for \emph{arbitrary signal spaces} and \emph{any strictly proper scoring rule}.

\subsection{General Setup}

\paragraph{Scoring Rules and Regret.}
Let $G:[0,1]\to\mathbb{R}$ be any strictly convex, continuously differentiable function. As established in  \Cref{sec:model}, the regret associated with a strictly proper scoring rule generated by $G$ is the Bregman divergence:
\[ B_G(p \| q) = G(p) - G(q) - (p-q)G'(q). \]
On the compact set $[0,1]^2$, $B_G$ is continuous and bounded.

\paragraph{Experiments.}
Let $\mathcal{S}$ be a finite set of signals. An \emph{experiment} is denoted by $\bp=(\pi_1,\pi_0)\in \Delta(\mathcal{S})^{\{0,1\}}$, where $\pi_\theta\in\Delta(\mathcal{S})$ is the signal distribution conditional on state $\theta$. Let $\Pi\subseteq\Delta(\mathcal{S})\times\Delta(\mathcal{S})$ be the set of feasible experiments. We assume there exists a unique uninformative signal distribution $\pi^0\in\Pi$, so the completely uninformative experiment is $\bp^0=(\pi^0,\pi^0)$. To ensure the problem is well-posed (i.e., high signals indicate state 1), we also assume that for any $(\pi_1, \pi_0)\in \Pi$:
\[ D_{KL}(\pi_1\|\pi_1')\le D_{KL}(\pi_1\|\pi_0') \quad \text{and} \quad D_{KL}(\pi_0\|\pi_0')\le D_{KL}(\pi_0\|\pi_1') \]
for any other $(\pi_1',\pi_0')\in \Pi$. This means under the true state, the signal distribution associated with that state is closer (in KL divergence) than the distribution associated with the opposite state.

\paragraph{The Game.}
With $n$ samples, Nature chooses a probability measure $\sigma_n$ on $\Pi$ with support $\Pi_n\subset \Pi$. Given a realized experiment $\bp=(\pi_1,\pi_0)$ and state $\theta$, the data $K=(K_s)_{s\in\mathcal{S}}$ follows a Multinomial distribution with likelihood $l(\pi_\theta;K) = \prod_{s \in \mathcal{S}} \pi_\theta(s)^{K_s}$.

The Oracle, knowing $\bp$, forms the posterior:
\[
q_{\text{oracle}}(K;\bp)\;=\;\frac{\mu\,l(\pi_1;K)}{\mu\,l(\pi_1;K)+(1-\mu)\,l(\pi_0;K)},
\]
with $\mu$ being the prior belief on state $\theta=1$.
The DM chooses a rule $a(K)$ to minimize ex-ante expected regret against $\sigma_n$. The minimax problem is the $[P]$ defined in  \Cref{subsec:example}:
\begin{equation*}
\min_{a(\cdot)} \max_{\sigma_n \in \Delta(\Pi)} \mathbb{E}_{\sigma_n, K} \left[ B_G\left(q_{\text{oracle}}(K; \bp) \parallel a(K)\right) \right]. \eqno{[P]}
\end{equation*}

\subsection{Main Result}

\begin{theorem}\label{thm:general_main}
In \textit{any} equilibrium, the support of Nature's strategy $\Pi_n$ must collapse to the uninformative experiment $\pi^0$ at the CLT rate:
        \begin{align*}
        \sup_{\bp \in \Pi_n} ||\pi_1 - \pi_0||_1 &= \Theta(1/\sqrt{n}).
        \end{align*}
The ex-ante equilibrium regret remains strictly positive and does not vanish as $n\to\infty$.
\end{theorem}

\begin{proof}
    See appendix \Cref{pf:thm:general_main}.
\end{proof}

This theorem confirms that the $1/\sqrt{n}$-rate derived in the MSE game is not an artifact of quadratic loss; it is the fundamental ``sweet spot'' for adversarial ambiguity under any proper scoring rule. 
Nature’s strategy of scaling precision at the rate $1/\sqrt{n}$ ensures that the signal-to-noise ratio remains constant even as $n \to \infty$. 



\subsection{Proof Intuition}

The formal proof is relegated to the Appendix. Here, we outline the heuristic argument showing why $\Theta(1/\sqrt{n})$ is the unique rate that sustains a non-trivial equilibrium.

\paragraph{Why not too informative? (The Upper Bound)}
Suppose Nature includes an experiment in which the signal strength is asymptotically larger than $1/\sqrt{n}$. 
More formally, we write $f(n)=\omega(g(n))$ if $f(n)/g(n)\to\infty$ as $n\to\infty$.
Thus, Nature includes an experiment where the signal strength is of order $\omega(1/\sqrt{n})$. This implies the KL-divergence is of order $\omega(1/n)$. Given this experiment is the realized true experiment $\bp_{true}$, the DM can statistically partition the support $\Pi_n$ into two sets based on the observed data:
\begin{itemize}
    \item \textbf{Distinguishable DGPs ($\mathcal{B}$):} Experiments that are statistically far from the truth.
    \[ \mathcal{B} = \left\{ \bp \in \Pi_n : D_{KL}(\bp_{true} \| \bp) = \omega(1/n) \right\} \]
    \item \textbf{Indistinguishable DGPs ($\mathcal{A}$):} Experiments close to the truth.
    \[ \mathcal{A} = \Pi_n \setminus \mathcal{B} \]
\end{itemize}
With large $n$, the likelihood of experiments in $\mathcal{B}$ vanishes relative to $\bp_{true}$. The DM effectively rules out $\mathcal{B}$ and forms beliefs based on $\mathcal{A}$. Since the experiments in $\mathcal{A}$ are highly informative (converging to the truth which reveals the state), the DM learns the state perfectly. Regret converges to zero.

\paragraph{Why not too uninformative? (The Lower Bound)}
Suppose Nature includes an experiment in which the signal strength is asymptotically smaller than $1/\sqrt{n}$. More formally, we write $f(n)=o(g(n))$ if $f(n)/g(n)\to 0$ as $n\to\infty$.
Thus, Nature includes an experiment where the signal strength is of order $o(1/\sqrt{n})$. This implies the KL-divergence is $o(1/n)$. By the Law of Large Numbers, the log-likelihood ratios converge to 0. Consequently, the Oracle's posterior converges to the prior $\mu$:
\[ q_{\text{oracle}}(K) \xrightarrow{p} \mu. \]
Anticipating this, the DM can simply ignore the data and report the prior, $a(K) = \mu$. Since both the Oracle and the DM report $\mu$, the Bregman divergence vanishes. Regret converges to zero.

Thus, to prevent the regret from vanishing, Nature must choose a rate where the signal is strong enough to move the Oracle's posterior away from the prior, but weak enough to prevent the DM from perfectly distinguishing the DGP from noise. The only rate satisfying this is $\Theta(1/\sqrt{n})$.

\section{Related Literature}

\subsection{Decision-Theoretic and Behavioral Approaches to Learning}

Our paper contributes to a growing literature that moves beyond standard Bayesian updating to model how individuals learn in complex environments. \smallskip

{\it (i) Learning under Ambiguity.}  Our framework is closely related to the literature on learning in the presence of ambiguity. Early foundations for this approach were established by \cite{wald1950statistical} under the maxmin principle and \cite{savage1951theory} via the minimax regret criterion. The subsequent decision-theoretic literature characterizes the axiomatic properties and testable implications of various belief-updating rules under ambiguity (e.g., \cite{marinacci2002learning, epstein2007learning, marinacci2019learning}; see \cite{siniscalchi2009vector} and \cite{gilboa2016ambiguity} for a survey). These papers, focusing mostly on settings with a single signal draw, theorize learning as a protective response to uncertainty, proposing rules where the decision-maker (DM) adopts a ``worst-case'' interpretation of the signal tailored to the specific decision at hand.

Similarly to our approach, \cite{frick2022learning} and \cite{reshidi2025asymptotic} study a DM learning from a large sample of signals with unknown precision. Unlike our framework, however, they both adopt a {\bf maximin utility criterion}, which typically focuses on the least informative interpretation of realized signals as the ``worst-case.'' To avoid trivial no-learning outcomes, both studies assume strictly informative signals with precision bounded away from zero. Under this assumption, the DM eventually achieves full learning as the data size increases in \cite{frick2022learning}.\footnote{Their objective is instead to compare the learning rates and degree of dynamic inconsistency under alternative updating rules.}

By contrast, and similarly to our results, asymptotic learning fails in \cite{reshidi2025asymptotic}. This failure arises because the DM in their model identifies a {\it separate} adversarial DGP for {\it each} realized sample of signals, allowing for non-identical precision across observations. This {\bf sample-dependence} of the DGP destroys signal-independence, thereby causing the Law of Large Numbers (LLN) to fail.\footnote{In their model, an ex ante sample distribution, even when it is well defined,   is {\it non-identical} and {\it non-independent} across signal draws. The latter results in an LLN failure. If they had required consistency across samples and thus ensured signal-independence, as we require in our model, complete asymptotic learning (the LLN) would have held, even with non-identical signals, by virtue of McDiarmid's inequalities. } Our ex-ante approach, however, requires the adversarial DGP (i.e., signal precision) to remain identical across different sample realizations. While this consistency requirement preserves signal-independence, asymptotic learning still fails in our framework because adversarial Nature can strategically select signal precisions that converge toward uninformative levels.  Crucially, our {\bf minimax regret} criterion prevents the trivial no-learning outcomes typical of maxmin models. It ensures that even in the worst case, learning occurs at the CLT rate, characterized by a limiting Normal distribution.

 Beyond consistency across signals, our paper differs from existing ambiguity models in the consistency of the decision rule. In much of the existing literature, the ``worst-case'' posterior is {\bf act-dependent}; an agent may interpret the same signal $K$ differently depending on whether they face choice $f$ or $g$. Under the standard linear preferences, the resulting adversarial beliefs are generically extreme---collapsing to a single state except in knife-edge cases. In contrast, we require that the interpretation rule be consistent across all potential decision problems. By committing ex-ante to a rule that minimizes regret across a portfolio of potential future tasks, our DM effectively faces a spectrum of possible cost thresholds. This task aggregation induces convex ex-ante preferences, resulting in stable, non-degenerate beliefs that provide a robust interpretation of information across diverse contexts.

\smallskip

 {\it (ii) Decision and Signal Independent Updating.}
Similar to our approach, a significant strand of literature develops non-Bayesian theories of updating that are independent of any particular decision problem or the DM's preferences on that decision \citep{perea2009model, ortoleva2012modeling, zhao2022pseudo, dominiak2023inertial, tang2024theory}. At a high level, these models stylize belief formation as a process of minimal revision or geometric projection, where the DM moves their prior the ``shortest distance'' to a new belief state consistent with new information \citep{zhao2022pseudo, ke2024learning} or shifts paradigms based on logical consistency and surprise thresholds \citep{ortoleva2012modeling, basu2019bayesian}.

While both we and this literature seek act-independent beliefs, we distinguish ourselves by replacing geometric economy or logical consistency with the guiding principle of robustness. Our updating rule is endogenously driven by ambiguity aversion, providing a minimax-optimal guarantee against a Nature that may strategically provide uninformative data.
\smallskip

 {\it (iii) Behavioral Biases and Cognitive Constraints.}
Our work also touches upon a behavioral literature exploring systematic departures from Bayesian rationality within Expected Utility settings, focusing on internal cognitive or psychological constraints \citep{epstein2006temptation, epstein2008non, bordalo2023memory}. These models frequently identify biases such as overreaction to vivid signals or conservatism—a tendency to update less than a Bayesian—resulting from exogenous factors like memory interference or self-control costs \citep{epstein2006temptation, bordalo2023memory}.

We provide a decision-theoretic foundation for this trait, which we term \textit{under-inference}, but derive it as an endogenous response to DGP ambiguity. Unlike behavioral models of cognitive failure, our ``asymptotic caution'' is a strategic equilibrium where the DM under-infers to guard against a worst-case scenario of degraded signal precision.

\subsection{Asymptotic Statistics and Local Alternatives}


 Our work relates to a prominent literature applying statistical decision theory to treatment choice problems   \citep{manski2004statistical, stoye2009minimax, hirano2009asymptotics, tetenov2012statistical, montielolea2024decision}, which utilizes the minimax regret criterion to resolve state ambiguity. The most fundamental departure of our work from this tradition lies in the assumptions on how the data-generating process behaves with sample size. In these existing studies, signal precision is a fixed parameter of the stationary distribution and does not vary with the same size. Hence, the Law of Large Numbers (LLN) ensures that as the sample size $n$ grows, the signal-to-noise ratio increases indefinitely, eventually revealing the true state with perfect clarity. Consequently, the statistical decision problem in this framework becomes trivial in the limit, and the global maximum regret vanishes.

To maintain a non-degenerate limit experiment, this literature exogenously applies a ``magnifying glass'' to the indifference point $\theta_0$  between decision alternatives. By adopting a local parameterization of the form $\theta_n = \theta_0 + h/\sqrt{n}$, as pioneered in the Local Asymptotic Normality (LAN) framework  \citep{lecam1960locally, van2000asymptotic}, researchers like \cite{hirano2009asymptotics} zoom in on the microscopic level where the state is so close to the decision boundary that it remains difficult to distinguish from zero despite the high volume of data. While this characterizes local trade-offs, the $1/\sqrt{n}$ rate acts as an exogenous lens used to create a non-trivial problem out of an otherwise transparent environment.

By contrast, in our framework, the signal precision is determined endogenously and adversarially by Nature and depends on the sample size $n$.  Hence, the LLN fails (recall  \Cref{co2} and more generally  \Cref{thm:general_main}), and the decision problem remains non-trivial at the macroscopic scale. We show that the $1/\sqrt{n}$ scaling emerges endogenously as the equilibrium strategy of a Nature that strategically degrades signal quality to match the sample size. We thus offer a dual framework for the local alternatives method. While standard LAN models keep the noise fixed and scale the state alternatives ($\Delta\theta \sim 1/\sqrt{n}$), our model keeps the state global and scales the noise. This duality demonstrates that the Gaussian shift experiment---usually viewed as a local approximation---is a global structural property of learning under strategic ambiguity.


Beyond the statistical scaling, a crucial distinction lies in the ``decision dependence'' of the resulting rules. In the existing literature, the payoff is defined with respect to a specific, typically binary, decision problem. This yields a linear regret function, which is minimized by a ``cutoff'' or ``threshold'' rule \citep{stoye2009minimax, tetenov2012statistical, hirano2009asymptotics}. Such rules generate beliefs that are extreme except at the cutoff, as they are tailored to the requirements of a single task. By contrast, we follow the logic of \cite{schervish1989general} and aggregate over a continuum of potential tasks. This aggregation induces a strictly convex regret function, which is optimized by stable, non-degenerate posteriors that apply to multiple decisions rather than a single decision.

\section{Conclusion}

This paper introduces a regret-minimization framework for a DM who may learn robustly---i.e., form posterior beliefs---when the process generating her data is uncertain. We model the problem as a zero-sum game: the DM chooses posterior beliefs to minimize ex-ante Bregman Divergence regret, while an adversarial Nature chooses the quality of data. Our analysis provides precise solutions for small sample sizes and reveals a key strategic insight in the large-sample limit: Nature's optimal strategy is to degrade the informativeness of the data at a specific rate, $1/\sqrt{n}$, as the sample size grows.

This leads to the failure of complete learning and failure of Law of Large Numbers against adversarial Nature, as the decision-maker's regret remains positive even with infinite data; and provides a novel, game-theoretic interpretation for the ``local alternatives'' framework in statistics.

By contrast, when the true DGP is fixed and informative, despite the misspecification, the DM eventually learns the state with enough data. However, robust learning is sub-exponential and therefore much slower than the Oracle’s exponential rate.  To guard against a worst-case scenario in which Nature has degraded the signal’s precision to nearly uninformative levels, the DM systematically under-infers from the data. These two results quantify the “price” of robustness.

More broadly, our results show that ambiguity can have a non-trivial effect on inference and learning, with implications for the design of robust statistical procedures and hypothesis testing. When the DGP is unknown, signals cannot be mapped unambiguously to likelihood ratios, making it unclear how to select the hypothesis.   In this case, one can use our model's robust beliefs to resolve the ambiguity and obtain a cutoff rule, much as in standard hypothesis testing.  Such a selection rule can be certified, based on our theory, to have a meaningful loss guarantee.

An immediate and promising avenue for future research is to extend this static, ex-ante framework to a dynamic setting. Investigating how a decision-maker would robustly update her beliefs as signals arrive sequentially would connect our approach to the rich literature on learning under ambiguity and dynamic consistency. Further extensions could also enrich the environment by considering non-binary states or alternative structures of ambiguity, deepening our understanding of how agents can learn effectively in an information-rich but fundamentally uncertain world.

\bibliographystyle{ACM-Reference-Format}
\bibliography{bibliography}

\appendix
\section{Proofs}

\subsection{Microfounding $G$ via Task Aggregation}
\label{app:schervish_microfoundation}

This section provides the formal justification for the interpretation of $G''(p)$ as the density of future decision problems. We describe how an ex-ante proper scoring rule arises from a portfolio of binary decisions, as pioneered by \cite{schervish1989general}.

\subsubsection{The Setup}
Consider a decision-maker (DM) who will eventually face a binary decision problem $(\theta, c)$. The DM must choose an action $x \in \{0, 1\}$. The state $\theta \in \{0, 1\}$ is unknown, and the DM holds a posterior belief $q = \Pr(\theta = 1)$. The parameter $c \in [0, 1]$ represents a cost threshold or indifference point. The payoffs are normalized such that the DM incurs a cost $c$ if she chooses $x=1$ when $\theta=0$ (Type I error), and a cost $1-c$ if she chooses $x=0$ when $\theta=1$ (Type II error).

For a fixed threshold $c$, the DM’s optimal action is $x=1$ if $q > c$ and $x=0$ otherwise. The minimized expected cost (the Bayes risk) for a single problem $c$ given belief $q$ is:
\begin{equation*}
    \ell(q, c) = \min \{ q(1-c), (1-q)c \}.
\end{equation*}

\subsubsection{Task Aggregation and the Generating Function}
Suppose the DM does not know which threshold $c$ she will face, but holds a measure $\lambda$ over the set of possible cost thresholds on $[0, 1]$. Her ex-ante expected cost across this portfolio of tasks is:
\begin{equation*}
    L(q) = \int_0^1 \ell(q, c) d\lambda(c).
\end{equation*}
Following \cite{schervish1989general}, any such $L(q)$ corresponds to a concave function. In our minimax-regret framework, the convex generating function $G(q)$ corresponds to the negative of this expected loss ($G = -L$), representing the ``value of information'' or the utility derived from the optimal decision.

\subsubsection{Weighting Representation}
The following theorem relates the density of the task distribution $\lambda$ to the curvature of the generating function $G$.

\begin{theorem}[Weighting Representation]
Let $G: [0, 1] \to \mathbb{R}$ be a strictly convex and twice-differentiable generating function for a proper scoring rule. There exists a unique non-negative measure $\lambda$ on $[0, 1]$ such that the curvature of $G$ at any point $p$ satisfies:
\begin{equation*}
    G''(p) = \lambda(p),
\end{equation*}
where $\lambda(p)$ is the density of the decision problems whose indifference threshold $c$ is exactly $p$.
\end{theorem}

\begin{proof}
The expected loss $L(q)$ can be written as:
\begin{equation*}
    L(q) = \int_0^q (1-q)c \, d\lambda(c) + \int_q^1 q(1-c) \, d\lambda(c).
\end{equation*}
Taking the derivative with respect to $q$ using the Leibniz rule:
\begin{equation*}
    L'(q) = -\int_0^q c \, d\lambda(c) + \int_q^1 (1-c) \, d\lambda(c).
\end{equation*}
Differentiating a second time yields:
\begin{equation*}
    L''(q) = -q \, \lambda(q) - (1-q) \, \lambda(q) = -\lambda(q).
\end{equation*}
Since $G(q) = -L(q)$ (plus potentially a linear term that disappears upon second differentiation), we obtain:
\begin{equation*}
    G''(q) = -L''(q) = \lambda(q).
\end{equation*}
Thus, the density of tasks at threshold $p$ is precisely the curvature of the generating function $G$ at that point.
\end{proof}

\subsection{Proof of  \Cref{prop1}}\label{pf:prop1}
\begin{proof}
The ex ante regret is:
\begin{align*}
    R(\bm{a},\pi)&=\sum_{k=0}^1 \Pr(k|\pi)(\Pr(\t=1|\pi,k)-a_k)^2\\
    &=\sum_{k=0}^1\frac{1}{2}(\Pr(\t=1|\pi,k)-a_k)^2\\
    &=\frac{1}{2}(\pi-a_1)^2+\frac{1}{2}(1-\pi-a_0)^2,
\end{align*}

Let $w$ be the weight that nature puts on $\pi=1$. We claim that $(a_1^*, a_0^*, w^*)=(\frac{3}{4},\frac{1}{4},\frac{1}{2})$ is the candidate solution. We verify that $w^*=\frac{1}{2}$ and $(a_0^*,a_1^*)$ form a saddle point. 

First, fixing $w^*=\frac{1}{2}$, 
\[wR(\bm{a},1)+(1-w)R(\bm{a},\frac{1}{2})=\frac{1}{4}(a_1-1)^2+\frac{1}{4}a_0^2+\frac{1}{4}(a_1-\frac{1}{2})^2 +\frac{1}{4}(a_0-\frac{1}{2})^2,\]
which is minimized at $a_1=\frac{3}{4}$, $a_0=\frac{1}{4}$.

Also, notice that $R(\bm{a}, \pi)$ is convex in $\pi$. This means that nature's optimal solution is to choose either $\pi=1$ and/or $\pi=\frac{1}{2}$. Fixing $(a_0^*, a_1^*)$, we have
\begin{align*}
    R(\bm{a},1)=R(\bm{a},\frac{1}{2}) \iff \frac{1}{2}(a_1-1)^2+\frac{1}{2}a_0^2 =\frac{1}{2}((a_1-\frac{1}{2})^2 +(a_0-\frac{1}{2})^2). 
\end{align*}
Thus, mixing between $\pi=\frac{1}{2}$ and $1$ is optimal for nature given the DM's action.
    \end{proof}

\subsection{Proof of  \Cref{thm:limit_eq}}\label{pf:thm:limit_eq}
\begin{proof}
We first show equation \Cref{eq:indifference} and \Cref{eq:localoptimality} can characterize the limit game's Nash equilibrium. Given nature's equilibrium strategy $\gamma$, DM's best response is derived via first order condition:
\begin{equation*}
a(z)=\frac{\E_{\gamma}\int_{-\infty}^{\infty}(\phi(z+2c^*)+\phi(z-2c^*))P(\theta|z,c^*) dz}{\E_{\gamma}\int_{-\infty}^{\infty}(\phi(z+2c^*)+\phi(z-2c^*)) dz}    
\end{equation*}

Let $\gamma(0)=1-w^*$, $\gamma(c^*)=w^*$,  nature's payoff when choosing $0$ is
\begin{equation*}
2\cdot\frac{1}{2}\int_{-\infty}^{\infty}\phi(z)\left(\frac{1}{2}-\frac{1-w^*+w^*e^{2c^*z-2{c^*}^2}}{2(1-w^*)+w^*(e^{2c^*z-2{c^*}^2}+e^{-2c^*z-2{c^*}^2})}\right)^2 dz    
\end{equation*}
Nature's payoff when choosing $c$ is
\begin{equation*}
\int_{-\infty}^{\infty} \frac{\left[ \phi\left(z-2c^*\right) + \phi\left(z+2c^*\right) \right]}{2} 
 \left[ \frac{(1-w^*)(e^{-2c^*z} - e^{2c^*z})}{\left[2(1-w^*) + w^*e^{-2{c^*}^2} \left(e^{2c^*z} + e^{-2c^*z}\right)\right]\left(e^{2c^*z} + e^{-2c^*z}\right)} \right]^2 dz    
\end{equation*}
which is exactly the same as $\tilde{R}(0)$ and $\tilde{R}(c^*)$. Thus nature's optimality condition can be captured by equation \Cref{eq:indifference} and \Cref{eq:localoptimality}.

Equation \Cref{eq:indifference} and \Cref{eq:localoptimality} can be written as 
\begin{align*}
F_1(c,w)&=\int_{-\infty}^{\infty} \frac{\left[ \phi\left(z-2c\right) + \phi\left(z+2c\right) \right]}{2} \cdot 
 \left[ \frac{(1-w)(e^{-2cz} - e^{2cz})}{\left[2(1-w) + we^{-2c^2} \left(e^{2cz} + e^{-2cz}\right)\right]\left(e^{2cz} + e^{-2cz}\right)} \right]^2 dz\\
 &-\int_{-\infty}^{\infty} \phi(z) \left(\frac{w(e^{2cz-2c^2}-e^{-2cz-2c^2})}{2[2(1-w)+w(e^{2cz-2c^2}+e^{-2cz-2c^2})]}\right)^2 dz\\
&=0
\end{align*}

\begin{align*}
F_2(c,w) = & \frac{d\tilde{R}(\frac{1}{2}+\frac{c}{\sqrt{n}})}{dc}\\=& \int_{-\infty}^{\infty}  (\frac{(1-w) + w e^{2c(z-c)}}{2(1-w) + w \left(e^{2c(z-c)} + e^{-2c(z+c)}\right)} - \frac{1}{1+e^{-4cz}}) \times \\
& \Bigg\{ \left[ (z-2c)\phi(z-2c) - (z+2c)\phi(z+2c) \right]\\
&(\frac{(1-w) + w e^{2c(z-c)}}{2(1-w) + w \left(e^{2c(z-c)} + e^{-2c(z+c)}\right)} - \frac{1}{1+e^{-4cz}}) \\
& \quad - 4z \left(\phi(z-2c) + \phi(z+2c)\right) \frac{e^{-4cz}}{(1+e^{-4cz})^2} \Bigg\} dz\\
&=0
\end{align*}

\subsubsection*{Existence}
\begin{lemma}
There exists a pair $(c^*,w^*) \in (0,1] \times [0,1]$ that solves
\begin{align*}
    F_1(c,w) &= 0, \\
    F_2(c,w) &= 0.
\end{align*}
\end{lemma}

\begin{proof}
    \textbf{Step 1:}  We show that for any $c \in (0,1]$, there exists some $w \in [0,1]$ such that $F_1(c,w)=0$. 
    
    We observe that for any $c \in (0,1]$, we have \begin{align*}
F_1(c,1)=-\int_{-\infty}^{\infty} \phi(z) \left(\frac{e^{2cz-2c^2}-e^{-2cz-2c^2}}{2(e^{2cz-2c^2}+e^{-2cz-2c^2})}\right)^2 dz<0   
\end{align*}
and
\begin{equation*}
F_1(c,0)=\int_{-\infty}^{\infty} \frac{\left[ \phi\left(z-2c\right) + \phi\left(z+2c\right) \right]}{2} \cdot 
 \left[ \frac{e^{-2cz} - e^{2cz}}{2\left(e^{2cz} + e^{-2cz}\right)} \right]^2 dz>0.
\end{equation*} Thus, the Intermediate Value Theorem ensures what we wanted.

\textbf{Step 2:} We show that the equation $F_1(c,w) = 0$ implicitly defines $w$ as a continuous function of $c$ on $(0,1]$. 

    \begin{lemma}\label{lemma9}
    For all $(c,w) \in (0,1] \times [0,1]$, we have $\frac{\partial F_1(c,w)}{\partial w} \ne 0$.
    \end{lemma}
\begin{proof}
    \begin{align*}
\frac{\partial F_1(c,w)}{\partial w} &= \int_{-\infty}^{\infty} \Bigg( \left[ \phi\left(z-2c\right) + \phi\left(z+2c\right) \right] \cdot \left[ \frac{-(1-w) \sinh(2cz)}{2(1-w)\cosh(2cz) + 2we^{-2c^2} \cosh^2(2cz)} \right] \\
&\quad\cdot \left[ \frac{e^{-2c^2} \sinh(2cz)}{2((1-w) + we^{-2c^2}  \cosh(2cz))^2} \right] \\
&\quad - \phi(z) \cdot 2 \left(\frac{e^{-2c^2} \sinh(2cz)}{2(1-w) + we^{-2c^2}  \left(e^{2cz} + e^{-2cz}\right)}\right) \cdot \left[ \frac{2e^{-2c^2} \sinh(2cz)}{\left[2(1-w) + 2we^{-2c^2}  \cosh(2cz)\right]^2} \right] \Bigg) dz
\end{align*}

Since the first term is negative and the second term is positive, the derivative is negative (with equality when $2cz=0$).

Thus we can show when $c>0$,
\begin{equation*}
\frac{\partial F_1(c,w)}{\partial w} \neq 0. 
\end{equation*}
\end{proof}

Then we establish a lemma that guarantees the global version of the implicit function theorem.\footnote{We believe this result is known, but we could not locate a reliable reference. Proof see appendix \Cref{pf:lemma4}.}

\begin{lemma}\label{lemma4}
Let $F(x,y)$ be continuously differentiable on $(-a,a) \times (-b,b)$, and suppose that $\frac{\partial F}{\partial y}(x,y) \ne 0$ for all $x \in (-a,a)$. If for any $x_0\in(-a,a)$, there exists some $y_0\in(-b,b)$ such that $F(x_0,y_0)=0$, then there exists a unique differentiable function $\varphi(x)$ defined on $(-a,a)$ such that $F(x, \varphi(x)) = 0$ for all $x$ in that interval.
\end{lemma}

By \Cref{lemma4}  and \Cref{lemma9}, the equation $F_1(c,w) = 0$ implicitly defines a continuous function $w = w_1(\pi) \in (0,1)$ with $c \in (0,1]$.

    \textbf{Step 3:} We show that there exists $c^* \in (0,1]$ such that $F_2(c^*, w_1(c^*)) = 0$.

    Define the function $J(c) := F_2(c, w_1(c))$. Since $F_2(c,w)$ is continuous, and $w_1(c)$ is continuous by Step 2. Hence, $J(c)$ is continuous on $(0,1]$.

We can compute endpoints:
\[
J(0.2) = F_2(0.2, w_1(0.2)) \approx 0.138429>0, 
\]
\[
J(1) = F_2(1, w_1(1)) \approx -0.041547<0.
\]

By the Intermediate Value Theorem, there exists $c^* \in (0.2,1)$ such that $J(c^*) = 0$. Therefore, the pair $(c, w) := (c^*, w_1(c^*))$ is a solution to the system.
\end{proof}

\begin{figure}[htbp]
    \centering
    \includegraphics[width=0.8\linewidth]{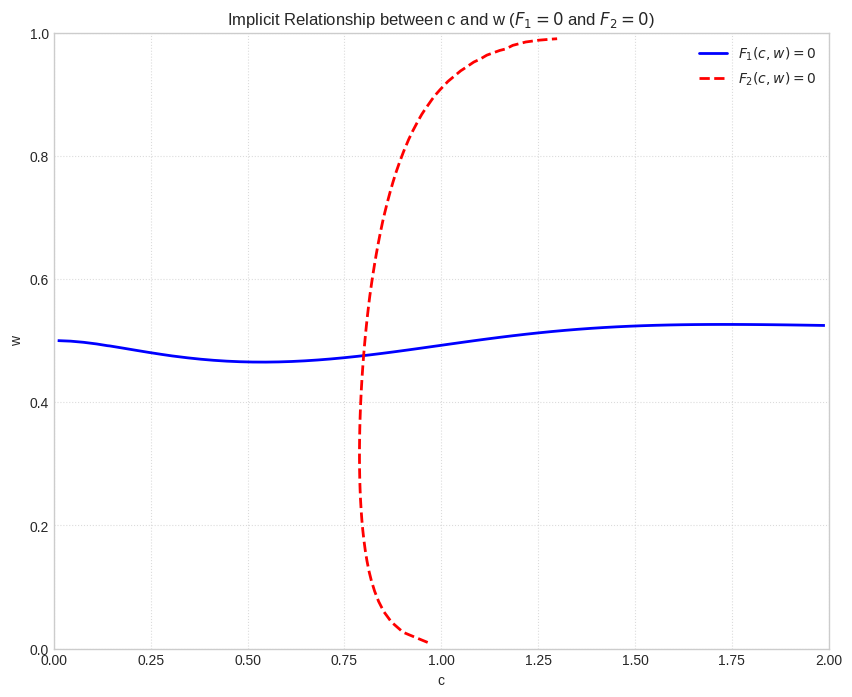}
    \label{fig:enter-label}
\end{figure}

\subsubsection*{Optimality}
First, it can be easily checked that the induced action $\bm{a}=q_{DM}$ is optimal for nature's strategy.  What remains to be shown is that given $\bm{a}$, the $(c^*,w^*)$ is optimal for the nature. We simply denote the limit regret $\tilde R(\bm a^*(c^*,w^*), c)$ as $\tilde R(c)$ in the following contents. Similarly, we also denote the finite regret $R_n(\bm a^*_n(c^*,w^*), c)$ as $\tilde R_n(c)$. 

Fixing DM's best response to $(c^*,w^*)$, we write 
\begin{align}
\tilde{R}(b)& \equiv \int_{-\infty}^\infty\frac{\left[ \phi\left(z-2b\right) + \phi\left(z+2b\right) \right]}{2} \left[\left(\frac{1-w^*+w^* e^{2c^*z-2{c^*}^2}}{2(1-w^*)+w^*(e^{2c^*z-2{c^*}^2}+e^{-2c^*z-2{c^*}^2})}-\frac{1}{1 + e^{-4bz}}\right)^2\right]dz\label{eq31}\\
&=\frac{1}{2}\int_{-\infty}^\infty\phi(z-2b)\left[\left(\frac{1-w^*+w^* e^{2c^*z-2c^2}}{2(1-w^*)+w^*(e^{2c^*z-2c^2}+e^{-2c^*z-2c^2})}-1\right)^2-\left(\frac{1}{1 + e^{-4bz}}-1\right)^2\right]dz\tag*{} \\
&+\frac{1}{2}\int_{-\infty}^\infty\phi(z+2b)\left[\left(\frac{1-w^*+w^* e^{2c^*z-2c^2}}{2(1-w^*)+w^*(e^{2c^*z-2c^2}+e^{-2c^*z-2c^2})}\right)^2-\left(\frac{1}{1 + e^{-4bz}}\right)^2\right]dz\label{eq32}\\
&=\int_{-\infty}^\infty\phi(z+2b) \left[\left(\frac{1-w^*+w^* e^{2c^*z-2{c^*}^2}}{2(1-w^*)+w^*(e^{2c^*z-2{c^*}^2}+e^{-2c^*z-2{c^*}^2})}\right)^2-\left(\frac{1}{1 + e^{-4bz}}\right)^2\right]dz\tag*{}
\end{align}
The first equality can be shown by arguing $\Cref{eq31}-\Cref{eq32}=0$.
\begin{align*}
    &\Cref{eq31}-\Cref{eq32}\\
    =&\int_{-\infty}^\infty \phi\left(z-2b\right)\left(\frac{1-w^*+w^* e^{2c^*z-2{c^*}^2}}{2(1-w^*)+w^*(e^{2c^*z-2{c^*}^2}+e^{-2c^*z-2{c^*}^2})}-\frac{1}{1 + e^{-4bz}}\right)\left(1-\frac{1}{1 + e^{-4bz}}\right)dz\\
    -&\int_{-\infty}^\infty \phi\left(z+2b\right)\left(\frac{1-w^*+w^* e^{2c^*z-2{c^*}^2}}{2(1-w^*)+w^*(e^{2c^*z-2{c^*}^2}+e^{-2c^*z-2{c^*}^2})}-\frac{1}{1 + e^{-4bz}}\right)\left(\frac{1}{1 + e^{-4bz}}\right)dz
\end{align*}
Observe that $\phi\left(z-2b\right)\left(1-\frac{1}{1 + e^{-4bz}}\right)=\phi\left(z+2b\right)\left(\frac{1}{1 + e^{-4bz}}\right)$. Thus, $\Cref{eq31}-\Cref{eq32}=0$. The second equality is by change of variable.

Namely, we fix the DM's action, and check if the nature has incentive to deviate.
\begin{align*}
\tilde{R}'(b) = \int_{-\infty}^{\infty} \Bigg\{ & -2(z+2b)\phi(z+2b) \left[ A(z)^2 - B(z,b)^2 \right] \\
& - 8z \, \phi(z+2b) \, \frac{e^{-4bz}}{\left(1+e^{-4bz}\right)^3} \Bigg\} \, dz\\
=\int_{-\infty}^{\infty}  & -2(z+2b)\phi(z+2b) \left[ A(z)^2 - B(z,b)^2 \right] dz
\end{align*}
where $A(z)=\frac{1-w^*+w^* e^{2c^*z-2{c^*}^2}}{2(1-w^*)+w^*(e^{2c^*z-2{c^*}^2}+e^{-2c^*z-2{c^*}^2})}$ and $B(z,b)=\frac{1}{1 + e^{-4bz}}$.

To prove optimality, it suffices to prove $\tilde{R}'(b)$ has at most 2 roots. We show this by arguing the auxiliary function $Q(b)=e^{2b^2}\tilde{R}'(b)$ is a strictly concave function and hence has at most 2 roots. That is, we want to show $Q''(b) <0$. 

\begin{align*}
        Q(b)=\frac{2}{\sqrt{2\pi}}\int_{-\infty}^\infty-(z+2b)e^{-\frac{1}{2}z^2-2bz}[A(z)^2-B(z,b)^2]dz
    \end{align*}
    \begin{align*}
        Q'(b)=\frac{4}{\sqrt{2\pi}}\int_{-\infty}^\infty\Bigg\{ &(z^2+2bz-1)e^{-\frac{1}{2}z^2-2bz}[A(z)^2-B(z,b)^2]\\
        &4z^2e^{-\frac{1}{2}z^2-2bz}\frac{e^{-4bz}}{\left(1 + e^{-4bz}\right)^3}\Bigg\}\,dz
    \end{align*}
    \begin{align*}
        Q''(b)=\frac{4}{\sqrt{2\pi}}\int_{-\infty}^\infty\Bigg\{ &-2z(z^2+2bz-2)e^{-\frac{1}{2}z^2-2bz}[A(z)^2-B(z,b)^2]\\
        &-16bz^2e^{-\frac{1}{2}z^2-2bz}\frac{e^{-4bz}}{\left(1 + e^{-4bz}\right)^3}\\
        &+16z^3e^{-\frac{1}{2}z^2-2bz}\frac{2e^{-8bz}-e^{-4bz}}{\left(1 + e^{-4bz}\right)^4}\Bigg\}\,dz\\
    =\frac{4}{\sqrt{2\pi}}\int_{-\infty}^\infty\Bigg\{ &-2z(z^2+2bz-2)e^{-\frac{1}{2}z^2-2bz}A(z)^2\\
        &+2z(z^2+2bz-2)e^{-\frac{1}{2}z^2-2bz}\frac{1}{\left(1 + e^{-4bz}\right)^3}\\
        &-12bz^2e^{-\frac{1}{2}z^2-2bz}\frac{e^{-4bz}}{\left(1 + e^{-4bz}\right)^3}\\
        &+16z^3e^{-\frac{1}{2}z^2-2bz}\frac{2e^{-8bz}-e^{-4bz}}{\left(1 + e^{-4bz}\right)^4}\Bigg\}\,dz
    \end{align*}

We can rewrite $Q''(b)$ as\footnote{In the calculations above and throughout the remainder of the paper, certain terms drop out because they are odd functions of $z$ and therefore integrate to zero over a symmetric domain.}:
\begin{align}
    Q''(b)=\frac{4}{\sqrt{2\pi}}\int_{-\infty}^\infty\Bigg\{ &-2z(z^2+2bz-2)e^{-\frac{1}{2}z^2-2bz}A(z)^2\tag{T1}\\
        &-4ze^{-\frac{1}{2}z^2-2bz}\frac{1}{\left(1 + e^{-4bz}\right)^3}\tag{T2}\\
        &-4bz^2e^{-\frac{1}{2}z^2-2bz}\frac{3e^{-4bz}-1}{\left(1 + e^{-4bz}\right)^3}\tag{T3}\\
        &+2z^3e^{-\frac{1}{2}z^2-2bz}\frac{8(2e^{-8bz}-e^{-4bz})+(1 + e^{-4bz})}{\left(1 + e^{-4bz}\right)^4}\Bigg\}\,dz\tag{T4}
\end{align}

Denote 
\begin{equation*}
T_1=\int_{-\infty}^{\infty} -2z(z^2+2bz-2)e^{-\frac{1}{2}z^2-2bz} A(z)^2 dz=  \int_{-\infty}^{\infty} f(b,z)A(z)^2 dz 
\end{equation*}

\begin{equation*}
T_2=-4\int_{-\infty}^{\infty}ze^{-\frac{1}{2}z^2-2bz}\frac{1}{\left(1 + e^{-4bz}\right)^3}dz
\end{equation*}

\begin{equation*}
T_3=-4b\int_{-\infty}^{\infty} z^2e^{-\frac{1}{2}z^2-2bz}\frac{3e^{-4bz}-1}{\left(1 + e^{-4bz}\right)^3} dz   
\end{equation*}

\begin{equation*}
T_4=2\int_{-\infty}^{\infty}z^3e^{-\frac{1}{2}z^2-2bz}\frac{8(2e^{-8bz}-e^{-4bz})+(1 + e^{-4bz})}{\left(1 + e^{-4bz}\right)^4}dz
\end{equation*}

where 
\begin{equation*}
f(b,z)= -2z(z^2+2bz-2)e^{-\frac{1}{2}z^2-2bz}   
\end{equation*}

\paragraph{Prove $T_1<0$ for every $c$, $w \in (0,1)$}
Notice that
\begin{equation*}
-2z(z^2+2bz-2)e^{-\frac{1}{2}z^2-2bz} =\frac{d}{dz} 2z^2 e^{-\frac{1}{2}z^2-2bz}     
\end{equation*}

Denote
\begin{equation*}
G(b,z)= 2z^2 e^{-\frac{1}{2}z^2-2bz}        
\end{equation*}

Notice there are three zero point for $f(b,z)$, from left to right we denote them by $b_1$, $0$, and $b_2$.

\begin{equation*}
G(b,-\infty)= \lim_{z \rightarrow -\infty} 2z^2 e^{-\frac{1}{2}z^2-2bz}=0    
\end{equation*}

\begin{equation*}
G(b,0)=0    
\end{equation*}

\begin{equation*}
G(b,\infty)= \lim_{z \rightarrow \infty} 2z^2 e^{-\frac{1}{2}z^2-2bz}=0       
\end{equation*}

Since
\begin{equation*}
\frac{d}{dz} A(z)=  \frac{d}{dz}  \frac{1-w+w e^{2cz-2c^2}}{2(1-w)+w(e^{2cz-2c^2}+e^{-2cz-2c^2})}>0
\end{equation*}

We have
\begin{align*}
&\int_{-\infty}^{b_1} -2z(z^2+2bz-2)e^{-\frac{1}{2}z^2-2bz} A(z)^2 dz\\
<& \int_{-\infty}^{b_1} -2z(z^2+2bz-2)e^{-\frac{1}{2}z^2-2bz} A(b_1)^2 dz\\
=&  (G(b,b_1)-G(b,-\infty))A(b_1)^2\\
=&  G(b,b_1)A(b_1)^2
\end{align*}

\begin{align*}
&\int_{b_1}^{0} -2z(z^2+2bz-2)e^{-\frac{1}{2}z^2-2bz} A(z)^2 dz\\
<& \int_{b_1}^{0} -2z(z^2+2bz-2)e^{-\frac{1}{2}z^2-2bz} A(b_1)^2 dz\\
=&  (G(b,0)-G(b,b_1)A(b_1)^2\\
=&  -G(b,b_1)A(b_1)^2    
\end{align*}

\begin{align*}
&\int_{0}^{b_2} -2z(z^2+2bz-2)e^{-\frac{1}{2}z^2-2bz} A(z)^2 dz\\
<& \int_{0}^{b_2} -2z(z^2+2bz-2)e^{-\frac{1}{2}z^2-2bz} A(b_2)^2 dz\\
=&  (G(b,b_2)-G(b,0))A(b_2)^2\\
=&  G(b,b_2)A(b_2)^2     
\end{align*}

\begin{align*}
&\int_{b_2}^{\infty} -2z(z^2+2bz-2)e^{-\frac{1}{2}z^2-2bz} A(z)^2 dz\\
<& \int_{b_2}^{\infty} -2z(z^2+2bz-2)e^{-\frac{1}{2}z^2-2bz} A(b_2)^2 dz\\
=&  (G(b,\infty)-G(b,b_2))A(b_2)^2\\
=&  -G(b,b_2)A(b_2)^2     
\end{align*}
   
Thus
\begin{align*}
T_1&=\int_{-\infty}^{\infty} -2z(z^2+2bz-2)e^{-\frac{1}{2}z^2-2bz} A(z)^2 dz=  \int_{-\infty}^{\infty} f(b,z)A(z)^2 dz \\
&<G(b,b_1)A(b_1)^2-G(b,b_1)A(b_1)^2+G(b,b_2)A(b_2)^2-G(b,b_2)A(b_2)^2\\
&=0
\end{align*}

\paragraph{Prove $T_2+T_3+T_4<0$}

\begin{align*}
T_2&=-4\int_{-\infty}^{\infty}ze^{-\frac{1}{2}z^2-2bz}\frac{1}{\left(1 + e^{-4bz}\right)^3}dz\\
&=-4\int_{0}^{\infty}ze^{-\frac{1}{2}z^2-2bz}\frac{1}{\left(1 + e^{-4bz}\right)^3}dz-4\int_{-\infty}^{0}ze^{-\frac{1}{2}z^2-2bz}\frac{1}{\left(1 + e^{-4bz}\right)^3}dz\\
&= -4\int_{0}^{\infty}ze^{-\frac{1}{2}z^2-2bz}\frac{1}{\left(1 + e^{-4bz}\right)^3}dz-4 \int_{-\infty}^{0}ze^{-\frac{1}{2}z^2-2bz}\frac{1}{\left(1 + e^{-4bz}\right)^3}dz\\
&=-4\int_{0}^{\infty}ze^{-\frac{1}{2}z^2-2bz}\frac{1}{\left(1 + e^{-4bz}\right)^3}dz-4\int_{\infty}^{0}ze^{-\frac{1}{2}z^2+2bz}\frac{1}{\left(1 + e^{4bz}\right)^3}dz\\
&=-4\int_{0}^{\infty}ze^{-\frac{1}{2}z^2-2bz}\frac{1}{\left(1 + e^{-4bz}\right)^3}-ze^{-\frac{1}{2}z^2+2bz}\frac{1}{\left(1 + e^{4bz}\right)^3}dz\\
&=-4\int_{0}^{\infty}ze^{-\frac{1}{2}z^2}\frac{1}{\left(1 + e^{-4bz}\right)^3}(e^{-2bz}-e^{-10bz})dz
\end{align*}

\begin{align*}
T_3&=-4b\int_{-\infty}^{\infty} z^2e^{-\frac{1}{2}z^2-2bz}\frac{3e^{-4bz}-1}{\left(1 + e^{-4bz}\right)^3} dz \\
&=-4b\int_{0}^{\infty} z^2e^{-\frac{1}{2}z^2}(\frac{3e^{-6bz}-e^{-2bz}}{\left(1 + e^{-4bz}\right)^3}+\frac{3e^{6bz}-e^{2bz}}{\left(1 + e^{4bz}\right)^3}) dz \\
&=-4b\int_{0}^{\infty} z^2e^{-\frac{1}{2}z^2}(\frac{3e^{-6bz}-e^{-2bz}}{\left(1 + e^{-4bz}\right)^3}+\frac{3e^{6bz}-e^{2bz}}{\left(1 + e^{4bz}\right)^3}) dz\\
&=-4b\int_{0}^{\infty} z^2e^{-\frac{1}{2}z^2}(\frac{3e^{-6bz}-e^{-2bz}}{\left(1 + e^{-4bz}\right)^3}+\frac{3e^{-6bz}-e^{-10bz}}{\left(1 + e^{-4bz}\right)^3}) dz\\
&=-4b\int_{0}^{\infty} z^2e^{-\frac{1}{2}z^2}(\frac{6e^{-6bz}-e^{-10bz}-e^{-2bz}}{\left(1 + e^{-4bz}\right)^3}) dz\\
\end{align*}

\begin{align*}
    T_4=&2\int_{-\infty}^{\infty}z^3e^{-\frac{1}{2}z^2-2bz}\frac{8(2e^{-8bz}-e^{-4bz})+(1 + e^{-4bz})}{\left(1 + e^{-4bz}\right)^4}dz\\
    =&2\int_{-\infty}^{\infty}z^3e^{-\frac{1}{2}z^2-2bz}\frac{16e^{-8bz}-7e^{-4bz}+1}{\left(1 + e^{-4bz}\right)^4}dz\\
    =&2\int_{0}^{\infty}z^3e^{-\frac{1}{2}z^2}(\frac{16e^{-10bz}-7e^{-6bz}+e^{-2bz}}{\left(1 + e^{-4bz}\right)^4}-\frac{16e^{10bz}-7e^{6bz}+e^{2bz}}{\left(1 + e^{4bz}\right)^4})dz\\
    =&2\int_{0}^{\infty}z^3e^{-\frac{1}{2}z^2}(\frac{16e^{-10bz}-7e^{-6bz}+e^{-2bz}}{\left(1 + e^{-4bz}\right)^4}-\frac{16e^{-6bz}-7e^{-10bz}+e^{-14bz}}{\left(1 + e^{-4bz}\right)^4})dz\\
    =&2\int_{0}^{\infty}z^3e^{-\frac{1}{2}z^2}(\frac{23e^{-10bz}-23e^{-6bz}+e^{-2bz}-e^{-14bz}}{\left(1 + e^{-4bz}\right)^4})dz
\end{align*}

Define an auxiliary function
\begin{equation*}
F(z)=-2z^2 e^{-\frac{1}{2}z^2} \frac{e^{2bz} - e^{-2bz}}{(e^{2bz} + e^{-2bz})^2}    
\end{equation*}
We have
\begin{equation*}
\frac{d}{dz} F(z)=-4ze^{-\frac{1}{2}z^2} \frac{e^{2bz} - e^{-2bz}}{(e^{2bz} + e^{-2bz})^2}-4b z^2 e^{-\frac{1}{2}z^2}\frac{6-e^{-4bz}-e^{4bz}}{(e^{2bz} + e^{-2bz})^3}+2z^3e^{-\frac{1}{2}z^2} \frac{e^{2bz} - e^{-2bz}}{(e^{2bz} + e^{-2bz})^2}
\end{equation*}

We have
\begin{align*}
&F(\infty)-F(0)= \int_{0}^{\infty}  \frac{d}{dz} F(z)dz\\
&=\int_{0}^{\infty} -4ze^{-\frac{1}{2}z^2} \frac{e^{2bz} - e^{-2bz}}{(e^{2bz} + e^{-2bz})^2}-4b z^2 e^{-\frac{1}{2}z^2}\frac{6-e^{-4bz}-e^{4bz}}{(e^{2bz} + e^{-2bz})^3}+2z^3e^{-\frac{1}{2}z^2} \frac{e^{2bz} - e^{-2bz}}{(e^{2bz} + e^{-2bz})^2} dz\\
&=\int_{0}^{\infty} -4ze^{-\frac{1}{2}z^2} \frac{e^{-2bz} - e^{-6bz}}{(1 + e^{-4bz})^2}-4b z^2 e^{-\frac{1}{2}z^2}\frac{6e^{-6bz}-e^{-10bz}-e^{-2bz}}{(1 + e^{-4bz})^3}+2z^3e^{-\frac{1}{2}z^2} \frac{e^{-2bz}-e^{-6bz}}{(1 + e^{-4bz})^2} dz\\
&=\int_{0}^{\infty} -4ze^{-\frac{1}{2}z^2} \frac{e^{-2bz} - e^{-10bz}}{(1 + e^{-4bz})^3}-4b z^2 e^{-\frac{1}{2}z^2}\frac{6e^{-6bz}-e^{-10bz}-e^{-2bz}}{(1 + e^{-4bz})^3}+2z^3e^{-\frac{1}{2}z^2} \frac{e^{-2bz}-e^{-6bz}}{(1 + e^{-4bz})^2} dz\\
\end{align*}

Notice that
\begin{equation*}
-48z^3 e^{-\frac{1}{2}z^2} \frac{e^{2bz} - e^{-2bz}}{(e^{2bz} + e^{-2bz})^4}+2z^3e^{-\frac{1}{2}z^2} \frac{e^{-2bz}-e^{-6bz}}{(1 + e^{-4bz})^2}=2z^3e^{-\frac{1}{2}z^2}(\frac{23e^{-10bz}-23e^{-6bz}+e^{-2bz}-e^{-14bz}}{\left(1 + e^{-4bz}\right)^4})    
\end{equation*}

Thus 
\begin{align*}
&T_2+T_3+T_4\\
=& -4\int_{0}^{\infty}ze^{-\frac{1}{2}z^2}\frac{1}{\left(1 + e^{-4bz}\right)^3}(e^{-2bz}-e^{-10bz})dz\\
&-4b\int_{0}^{\infty} z^2e^{-\frac{1}{2}z^2}(\frac{6e^{-6bz}-e^{-10bz}-e^{-2bz}}{\left(1 + e^{-4bz}\right)^3}) dz+ 2\int_{0}^{\infty}z^3e^{-\frac{1}{2}z^2}(\frac{23e^{-10bz}-23e^{-6bz}+e^{-2bz}-e^{-14bz}}{\left(1 + e^{-4bz}\right)^4})dz \\
=& -4\int_{0}^{\infty}ze^{-\frac{1}{2}z^2}\frac{1}{\left(1 + e^{-4bz}\right)^3}(e^{-2bz}-e^{-10bz})dz\\
&-4b\int_{0}^{\infty} z^2e^{-\frac{1}{2}z^2}(\frac{6e^{-6bz}-e^{-10bz}-e^{-2bz}}{\left(1 + e^{-4bz}\right)^3}) dz\\
&+ \int_{0}^{\infty}-48z^3 e^{-\frac{1}{2}z^2} \frac{e^{2bz} - e^{-2bz}}{(e^{2bz} + e^{-2bz})^4}+2z^3e^{-\frac{1}{2}z^2} \frac{e^{-2bz}-e^{-6bz}}{(1 + e^{-4bz})^2} dz \\
=& -4\int_{0}^{\infty}ze^{-\frac{1}{2}z^2}\frac{1}{\left(1 + e^{-4bz}\right)^3}(e^{-2bz}-e^{-10bz})dz\\
&-4b\int_{0}^{\infty} z^2e^{-\frac{1}{2}z^2}(\frac{6e^{-6bz}-e^{-10bz}-e^{-2bz}}{\left(1 + e^{-4bz}\right)^3}) dz\\
&+ \int_{0}^{\infty}-48z^3 e^{-\frac{1}{2}z^2} \frac{e^{2bz} - e^{-2bz}}{(e^{2bz} + e^{-2bz})^4}dz+ \int_{0}^{\infty}2z^3e^{-\frac{1}{2}z^2} \frac{e^{-2bz}-e^{-6bz}}{(1 + e^{-4bz})^2} dz \\
=& F(\infty)-F(0) + \int_{0}^{\infty}-48z^3 e^{-\frac{1}{2}z^2} \frac{e^{2bz} - e^{-2bz}}{(e^{2bz} + e^{-2bz})^4}dz
\end{align*}

Since
\begin{equation*}
\int_{0}^{\infty}-48z^3 e^{-\frac{1}{2}z^2} \frac{e^{2bz} - e^{-2bz}}{(e^{2bz} + e^{-2bz})^4}dz<0    
\end{equation*}

\begin{equation*}
F(\infty)= \lim_{z \rightarrow \infty}  -2z^2 e^{-\frac{1}{2}z^2} \frac{e^{2bz} - e^{-2bz}}{(e^{2bz} + e^{-2bz})^2}=0  
\end{equation*}

\begin{equation*}
F(0)=0    
\end{equation*}

We can show
\begin{equation*}
T_2+T_3+T_4<0
\end{equation*}

Combine with $T_1<0$, we have
\begin{equation*}
 Q''(b)=T_1+ T_2+T_3+T_4<0   
\end{equation*}

Thus, $Q(b)=e^{2b^2}\tilde{R}'(b)$ is strictly concave, which implies that it has at most 2 roots. That means $\tilde{R}'(b)$ has at most 2 roots and $\tilde{R}(b)$ has at most 2 critical points. Moreover, 
    \begin{align*}
        \tilde{R}'(0)&=\int_{-\infty}^{\infty}-2z\phi(z)[A(z)^2-\frac{1}{4}] dz\\
        &=\int_0^{\infty}-2z\phi(z)\bigg(\frac{1-w^*+w^*(e^{2c^*z-2{c^*}^2}-e^{-2c^*z-2{c^*}^2})}{2(1-w^*)+w^*(e^{2c^*z-2{c^*}^2}+e^{-2c^*z-2{c^*}^2})}\bigg)^2 dz\\
        &<0
    \end{align*}

Our first indifference condition says that $\tilde{R}(0)=\tilde{R}(c^*)$. Together with $\tilde{R}(b)$ decreasing at $b=0$, $\tilde{R}(b)$ must increase from some point in $(0,c^*)$ to $c^*$. Thus, there must be at least one critical point in $(0,c^*)$.

Moreover, the second condition says that $c^*$ is a critical point for $\tilde{R}(b)$. Since there are at most two critical points, then $\tilde{R}(b)$ must be first decreasing, then increasing to $\tilde{R}(c^*)$, and then decreasing.

Thus, $\tilde{R}(0)=\tilde{R}(c^*)$ achieves the maximum.

For the simulated $(c^*,w^*)=(0.799, 0.476)$, 
\begin{figure}[htbp!]
    \centering
    \includegraphics[width=0.6\linewidth]{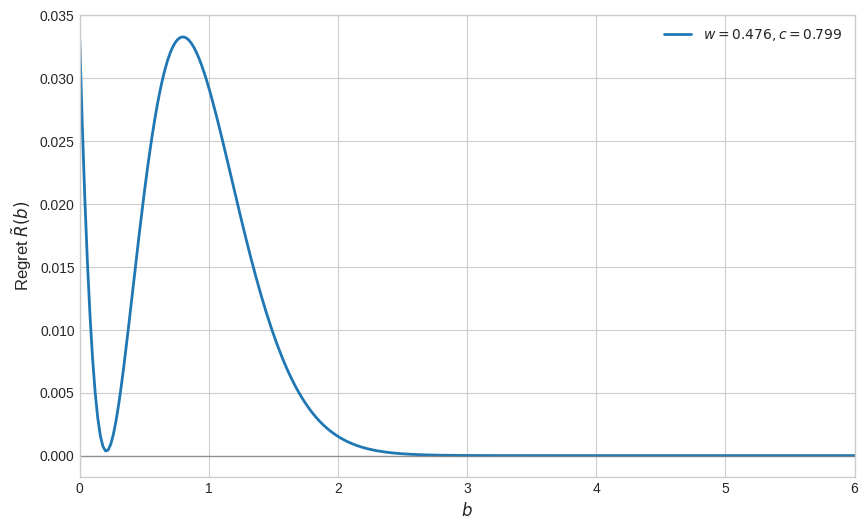}
\end{figure}

\subsubsection*{Uniqueness of $(c^*,w^*)$}
The proof proceeds by contradiction in two stages. First, we prove the uniqueness of the optimal interior signal quality $c^*$. Second, given the unique $c^*$, we prove the uniqueness of the randomization weight $w^*$.

Let $(\bm a, \gamma)$ denote a saddle point (Nash Equilibrium) of the limit game $\tilde{\Gamma}$, where $\bm a$ is the DM's strategy and $\gamma$ is Nature's strategy.

\vspace{\baselineskip}
\noindent\textbf{Step 1: Uniqueness of $c^*$}

\noindent Assume that there are two distinct equilibrium strategies for Nature, characterized by pairs $(c_1, w_1)$ and $(c_2, w_2)$ with $c_1 \neq c_2$. Let the full saddle points be $(\bm a_1, \gamma_1)$ and $(\bm a_2, \gamma_2)$. Thus, $\gamma_1$ is a probability distribution on $\{0, c_1\}$ with weight $w_1$ on $c_1$ and $\gamma_2$ is a probability distribution on $\{0, c_2\}$ with weight $w_2$ on $c_2$. $\bm a_1$ and $\bm a_2$ are the DM's best-response strategy to $\gamma_1$, $\gamma_2$, respectively.

Minimax theorem implies that if $(\bm a_1, \gamma_1)$ and $(\bm a_2, \gamma_2)$ are saddle points, then the ``crossed'' profiles $(\bm a_1, \gamma_2)$ and $(\bm a_2, \gamma_1)$ must also be saddle points.

Let us analyze the implication of $(\bm a_1, \gamma_2)$ being a saddle point. This requires that $\gamma_2$ be a best response for Nature to the DM's strategy $\bm a_1$. For the mixed strategy $\gamma_2$ to be a best response, every pure strategy in its support must also be a best response. This means that both $c=0$ and $c=c_2$ must maximize the regret function $c \mapsto \tilde{R}(\bm a_1, c)$.

This implies $c_2$ must be a global maximizer of $\tilde{R}(\bm a_1, c)$ over $c \ge 0$. The original saddle point $(\bm a_1, \gamma_1)$ also implies $c_1$ must be a global maximizer as well.

As established in the proof of optimality part (via the analysis of the auxiliary function $Q(b)$ which showed it to be strictly concave), the regret function $c \mapsto \tilde{R}(\bm a_1, c)$ has a unique global maximizer for $c > 0$, which is $c_1$. Therefore, it must be that $c_1 = c_2$.

\vspace{\baselineskip}
\noindent\textbf{Step 2: Uniqueness of $w^*$}

\noindent Now that we have established that $c^*$ is unique, we show that the corresponding weight $w^*$ must also be unique. Assume for contradiction that there are two distinct equilibria characterized by $(c^*, w_1)$ and $(c^*, w_2)$ with $w_1 \neq w_2$.

The equilibrium weight $w$ is determined as the solution to Nature's indifference condition, where the DM's action $\bm a(c^*, w)$ is itself a function of $w$:
\[
\tilde{R}(\bm a(c^*, w), c^*) = \tilde{R}(\bm a(c^*, w), 0)
\]
Let us define the function $F: [0,1] \to \mathbb{R}$ as:
\[
F(w) = \tilde{R}(\bm a(c^*, w), c^*) - \tilde{R}(\bm a(c^*, w), 0)
\]
An equilibrium weight $w^*$ must be a root of the equation $F(w) = 0$. In the existence part of proof for the limit game (\Cref{lemma9}), it was shown that the partial derivative of the indifference equation with respect to $w$ is non-zero. Specifically, it was established that $F(w)$ is a strictly decreasing function of $w$ for $w \in (0,1)$.

 Thus, there can only be a single value $w^* \in (0,1)$ that satisfies the indifference condition $F(w^*) = 0$. This implies it must be that $w_1 = w_2$.

Therefore, the equilibrium strategy for Nature, characterized by $(c^*, w^*)$, is unique.
\end{proof}

\subsection{Proof of  \Cref{thm2}}\label{pf:thm2}
\begin{proof}
We show that any sequence of Nash equilibria $(\bm a_n^*,\gamma_n^*)$ of the finite games $\Gamma_n$ converges (under the embeddings described below) to the unique Nash equilibrium $(\bm a^*,\gamma^*)$ of the limit game $\tilde\Gamma$.

\subsubsection{Embedding and Compactness}
Because the finite and limit games are defined on different strategy spaces, we first place them on the same compact domain. For Nature, recall that in $\Gamma_n$ she chooses a probability measure $\gamma_n$ on $[0,\sqrt n/2]$ over the local precision parameter. The following result shows we may restrict attention without loss to a large compact interval $[0,C]$ (high precision is strictly dominated for Nature).

\begin{lemma}
    There exists a $C>0$ such that the tail region $[C,\infty)$ is strictly dominated. Thus, it is without loss to restrict attention to precision in the interval $c\in[0,C]$.
\end{lemma}

\begin{proof}
We show there exists a $C>0$ such that the tail region $[C,\infty)$ is strictly dominated, and it is without loss to restrict attention
to $c\in[0,C]$.

Notice that for any measurable rule $\bm a(\cdot)$ and any $c\ge 0$,
\begin{align*}
\tilde{R}(\bm a,c)
&=\frac12\int_{-\infty}^{\infty}\phi(z+2c)\,\bigl(P(\theta\mid z,c)-\bm a(z)\bigr)^2\,dz
+\frac12\int_{-\infty}^{\infty}\phi(z-2c)\,\bigl(P(\theta\mid z,c)-\bm a(z)\bigr)^2\,dz\\
&\le \frac12\int_{-\infty}^{\infty}\phi(z+2c)\,dz+\frac12\int_{-\infty}^{\infty}\phi(z-2c)\,dz,
\end{align*}
where the inequality uses $0\le P(\theta\mid z,c)\le 1$ and $0\le \bm a(z)\le 1$, hence
$\bigl(P(\theta\mid z,c)-\bm a(z)\bigr)^2\le 1$ pointwise.

By the change of variables $u=z\pm 2c$,
\[
\int_{-\infty}^{\infty}\phi(z+2c)\,dz=\int_{-\infty}^{\infty}\phi(u)\,du=1,
\qquad
\int_{-\infty}^{\infty}\phi(z-2c)\,dz=\int_{-\infty}^{\infty}\phi(u)\,du=1.
\]
Thus $\tilde{R}(\bm a,c)\le 1$ for all $c$. Moreover, since $\phi$ is integrable and $\phi(z\pm 2c)\to 0$
pointwise as $c\to\infty$, dominated convergence yields
\[
\int_{-\infty}^{\infty}\phi(z+2c)\,\bigl(P(\theta\mid z,c)-\bm a(z)\bigr)^2\,dz \rightarrow 0,
\qquad
\int_{-\infty}^{\infty}\phi(z-2c)\,\bigl(P(\theta\mid z,c)-\bm a(z)\bigr)^2\,dz \rightarrow 0,
\]
and therefore
\[
\lim_{c\to\infty}\tilde{R}(\bm a,c)=0.
\]

Since $\tilde{R}(\bm a,\cdot)$ is nonnegative and continuous in $c$, and satisfies
$\tilde{R}(\bm a,c)\to 0$ as $c\to\infty$, there exists $C>0$ such that
\[
\sup_{c\ge C}\tilde{R}(\bm a,c)\le \varepsilon
\]
for any prescribed $\varepsilon>0$ by choosing $C$ large enough.
Hence the tail region $[C,\infty)$ is strictly dominated, and it is without loss to restrict attention
to $c\in[0,C]$.  \end{proof}

We therefore view each $\gamma_n$ as an element of $\Sigma:=\Delta([0,C])$ by restricting/identifying its mass on $[0,C]$. Equipped with the topology of weak convergence (e.g., the Prohorov metric), $\Sigma$ is compact and convex.

For DM's strategy, define a strategy space $\mathcal{A}$ as the set of \emph{monotone non-decreasing functions} $f: [0,1] \to [0,1]$. We embed the finite game's strategy space $A_n$ and the limit game's strategy space $A$ into this strategy space $\mathcal{A}$. For \emph{limit strategies}, given $\bm a:\mathbb R\to[0,1]$, define its embedding $\tilde{\bm a}\in\mathcal A$ by $\tilde{\bm a}(y)=\bm a(\Phi^{-1}(y))$ for $y\in[0,1]$, where $\Phi$ is the standard normal CDF. For \emph{finite strategies}, given $\bm a_n=(a_0,\ldots,a_n)\in A_n$, map each signal count $k$ to its asymptotic quantile under the uninformative experiment $c=0$ (equivalently $\pi=1/2$) via $y_{n,k}=\Phi\left(\frac{k - n/2}{\sqrt{n}/2}\right)$, and define $\tilde{\bm a}_n\in\mathcal A$ as the piecewise-linear interpolation of the points $\{(y_{n,k},a_k)\}_{k=0}^n$.
By Helly's Selection Theorem, $\mathcal{A}$ is compact under pointwise convergence. 

Moreover, the following result establishes the existence of saddle points in both the Finite Games and the Limit Game.

\begin{proposition}
    Saddle point exists in both Finite and Limit Games.
\end{proposition}

\begin{proof}
    For any Finite-$n$ Game, the strategy spaces are compact and convex. The regret function $R_n(\bm{\tilde{a}}, \gamma)$ is linear in $\gamma$ and strictly convex in $\bm{\tilde{a}}$. Hence, Sion' Minimax Theorem implies that saddle point exists in the Finite Games. 

    For the Limit Game, existence is established in \Cref{thm:limit_eq}.
\end{proof}

\subsubsection{Pointwise Convergence of Payoffs.}
Fix $(\tilde{\bm a},\gamma)\in\mathcal A\times\Sigma$. We show finite-game regret converges to the limit-game regret:
\[
R_n(\tilde{\bm a},\gamma)\;\longrightarrow\;\tilde R(\tilde{\bm a},\gamma).
\]

\begin{proposition}\label{prop:limit regret} Fix $(\tilde{\bm a},\gamma)\in\mathcal A\times\Sigma$. The finite-game regret pointwise converges to the limit-game regret:
\[
R_n(\tilde{\bm a},\gamma)\;\longrightarrow\;\tilde R(\tilde{\bm a},\gamma).
\]
\end{proposition}

\begin{proof}
Recall that the finite-game regret is given by:
\[
R_n(\tilde{\bm a}, \gamma) = \int_{[0,\sqrt{n}]} R_n(\tilde{\bm a}, c) \, \gamma(dc),
\]
where $R_n(\tilde{\bm a}, c)$ is the regret conditional on the local parameter $c$ (corresponding to precision $\pi_n = 1/2 + c/\sqrt{n}$). Since the loss function is bounded in $[0,1]$, we have $0 \le R_n(\tilde{\bm a}, c) \le 1$ for all $n$ and $c$. By the Dominated Convergence Theorem, it suffices to establish pointwise convergence in $c$. That is, we show that for any fixed $c \in \mathbb{R}_{+}$:
\[
\lim_{n \to \infty} R_n(\tilde{\bm a}, c) = \tilde R(\tilde{\bm a}, c).
\]
In the remainder of the proof, fix $c$. We analyze the convergence of the regret conditional on state $\theta=1$; the argument for $\theta=0$ is symmetric. The $\theta=1$ component of the limit regret is:
\[
\tilde{R}(\tilde{\bm a}, c | \theta=1) = \int_{-\infty}^{\infty} \phi(z-2c) \left( \frac{1}{1+e^{-4cz}} - \tilde{\bm a}(z) \right)^2 dz.
\]

\paragraph{Step 1: Grid Discretization and Local CLT}
Let $K$ denote the number of high signals. We standardize our grid using the variable:
\[
x_{n,K} \equiv \frac{2K - n}{\sqrt{n}}, \quad \text{with grid spacing } \Delta_n \equiv \frac{2}{\sqrt{n}}.
\]
We analyze the behavior of $K$ under the sequence of priors $\pi_n = 1/2 + c/\sqrt{n}$. We define the standardized random variable:
\[
Z_n = \frac{K - n\pi_n}{\sqrt{n\pi_n(1-\pi_n)}}.
\]
By the De Moivre-Laplace theorem, $Z_n$ converges in distribution to a standard normal $Z \sim N(0,1)$. To relate this to our grid $x_{n,k}$, we expand $Z_n$:
\begin{align*}
Z_n(K) &= \frac{K - n(1/2 + c/\sqrt{n})}{\sqrt{n(1/2 + c/\sqrt{n})(1/2 - c/\sqrt{n})}} \\
&= \frac{(K - n/2) - c\sqrt{n}}{\sqrt{n/4 - c^2}} \\
&= \frac{\frac{\sqrt{n}}{2} x_{n,K} - c\sqrt{n}}{\frac{\sqrt{n}}{2}\sqrt{1 - 4c^2/n}} \\
&= \frac{x_{n,K} - 2c}{\sqrt{1 - 4c^2/n}}.
\end{align*}
As $n \to \infty$, the denominator approaches 1 uniformly on compact sets. Thus, $Z_n(K=k) \rightarrow x_{n,k} - 2c$.
Since the probability mass function of the standardized binomial is approximated by $\frac{1}{\sigma_n}\phi(Z_n)$, and the standard deviation is $\sigma_n = \sqrt{n\pi_n(1-\pi_n)} = \frac{\sqrt{n}}{2}+O(\frac{1}{\sqrt{n}}) = \frac{1}{\Delta_n}$, we obtain the local limit form:
\begin{equation} \label{eq:LLT}
\Pr(K=k \mid c, \theta=1) = \Delta_n \phi(x_{n,k} - 2c) + o(\Delta_n).
\end{equation}

\paragraph{Step 2: Posterior Convergence}
Let $p_n(k, c) \equiv \Pr(\theta=1 \mid K=k, c)$. The log-likelihood ratio for the finite model is:
\begin{align*}
\log \left( \frac{p_n}{1-p_n} \right) &= (2k-n) \log \left( \frac{1/2 + c/\sqrt{n}}{1/2 - c/\sqrt{n}} \right) \\
&= \sqrt{n} x_{n,k} \cdot \log \left( \frac{1 + 2c/\sqrt{n}}{1 - 2c/\sqrt{n}} \right).
\end{align*}
Using the Taylor expansion, the term on the right becomes:
\[
\sqrt{n} x_{n,k} \left( \frac{4c}{\sqrt{n}} + O(n^{-3/2}) \right) = 4c x_{n,k} + O(n^{-1}).
\]
Thus, for any sequence $k_n$ such that $x_{n,k_n} \to z$, the posterior converges to the limit posterior:
\begin{equation} \label{eq:post_conv}
p_n(k_n, c) \xrightarrow{} \frac{1}{1+e^{-4cz}} \equiv p(z, c).
\end{equation}

\paragraph{Step 3: Decomposition and Convergence}
We decompose the sum over $k$ into a central region bounded by $M$ and the tails. Let $I_n(M) = \{k : |x_{n,k}| \le M\}$.
\[
R_n(\tilde{\bm a}, c | \theta=1) = \sum_{k \in I_n(M)} \Pr(K=k) (p_n(k,c) - \tilde{\bm a}(x_{n,k}))^2 \quad + \sum_{k \notin I_n(M)} \Pr(K=k) (p_n(k,c) - \tilde{\bm a}(x_{n,k}))^2.
\]

\textbf{The Tails:} Since the loss is bounded by 1, the tail sum is bounded by $\Pr(|X_{n,K}| > M)$. By the Central Limit Theorem, as $n \to \infty$, the distribution of $X_{n,K}$ converges to $N(2c, 1)$. Thus:
\[
\limsup_{n \to \infty} \sum_{k \notin I_n(M)} \Pr(K=k) (p_n(k,c) - \tilde{\bm a}(x_{n,k}))^2 \le 1 - \int_{-M}^M \phi(z-2c) dz.
\]
This term vanishes as $M \to \infty$.

\textbf{The Central Region:} Substituting the Local Limit approximation \Cref{eq:LLT} into the central sum:
\[
 \sum_{k \in I_n(M)} \left( \Delta_n \phi(x_{n,k} - 2c) + o(\Delta_n) \right) (p_n(k,c) - \tilde{\bm a}(x_{n,k}))^2
\]
Since $\tilde{\bm a}$ is a monotonic function, it is continuous almost everywhere. At any continuity point $z$ of $\tilde{\bm a}$, the discrete loss $(p_n(k,c) - \tilde{\bm a}(x_{n,k}))^2$ converges pointwise to the limit loss $(p(z,c) - \tilde{\bm a}(z))^2$ by \Cref{eq:post_conv}.
The sum is a Riemann sum approximating the integral of a bounded, almost-everywhere continuous function, such function is Riemann integrable, so the sum converges. Thus, as $n \to \infty$:
\[
\sum_{k \in I_n(M)} \left( \Delta_n \phi(x_{n,k} - 2c) + o(\Delta_n) \right) (p_n(k,c) - \tilde{\bm a}(x_{n,k}))^2 \longrightarrow \int_{-M}^M \phi(z-2c) \left( p(z,c) - \tilde{\bm a}(z) \right)^2 dz.
\]

\paragraph{Step 4: Conclusion}
Combining the bounds, we take limits in the order $n \to \infty$ followed by $M \to \infty$:
\begin{align*}
\lim_{n \to \infty} R_n(\tilde{\bm a}, c | \theta=1) &= \int_{-\infty}^{\infty} \phi(z-2c) \left( p(z,c) - \tilde{\bm a}(z) \right)^2 dz.
\end{align*}
Applying the symmetric argument for $\theta=0$ yields the full limit regret $\tilde{R}(\tilde{\bm a}, c)$. This completes the proof.
\end{proof}

\subsubsection{Uniform Convergence in the DM's Problem.}
For each $\gamma\in\Sigma$, define Nature's value function in the finite and limit games by
\[
V_n(\gamma):=\inf_{\tilde{\bm a}_n\in\mathcal A} R_n(\tilde{\bm a}_n,\gamma),
\qquad
V_\infty(\gamma):=\inf_{\tilde{\bm a}\in\mathcal A} \tilde R(\tilde{\bm a},\gamma).
\]
We claim that for any sequence $\gamma_n\to\gamma$ in $\Sigma$,
\begin{equation}\label{eq:value-pointwise}
V_n(\gamma_n)\;\longrightarrow\;V_\infty(\gamma).
\end{equation}
To justify the interchange of $\lim$ and $\inf$, it suffices to show that $R_n(\cdot,\gamma_n)$ converges to $\tilde R(\cdot,\gamma)$ uniformly on $\mathcal A$. This uniformity follows from a standard convex-analytic fact: on a compact convex set, pointwise convergence of convex functions to a continuous limit implies uniform convergence. Here, for each $n$, the map $\tilde{\bm a}_n\mapsto R_n(\tilde{\bm a}_n,\gamma_n)$ is convex because it is an expectation of quadratic losses in $\tilde{\bm a}_n$, and the limit map $\tilde{\bm a}\mapsto \tilde R(\tilde{\bm a},\gamma)$ is continuous (it is an integral of a squared difference against a fixed dominating density). Since we have pointwise convergence result above, uniform convergence on $\mathcal A$ follows, and \Cref{eq:value-pointwise} holds.

\subsubsection{Uniform Convergence of Value Functions on $\Sigma$.}
Next we strengthen \Cref{eq:value-pointwise} to uniform convergence of $V_n$ to $V_\infty$ on the compact set $\Sigma$:
\begin{equation}\label{eq:value-uniform}
\sup_{\gamma\in\Sigma}\bigl|V_n(\gamma)-V_\infty(\gamma)\bigr|\;\longrightarrow\;0.
\end{equation}
To see this, note first that for each $n$, $V_n$ is concave on $\Sigma$ because $R_n(\tilde{\bm a}_n,\gamma)$ is linear in $\gamma$ and $V_n$ is the pointwise infimum over $\tilde{\bm a}_n$ of linear functions. The same is true for $V_\infty$. Moreover, by the previous paragraph, $V_n(\gamma)\to V_\infty(\gamma)$ pointwise for every $\gamma\in\Sigma$. Since $\Sigma$ is compact and convex and $V_\infty$ is continuous (finite and concave on a compact convex domain), the same compactness/convexity argument implies that pointwise convergence of concave functions to a continuous limit is uniform, yielding \Cref{eq:value-uniform}.

\subsubsection{Convergence of Equilibrium Strategies.}
Consider Nature's equilibrium problem $\max_{\gamma\in\Sigma} V_n(\gamma)$. By \Cref{eq:value-uniform}, $V_n\to V_\infty$ uniformly on the compact set $\Sigma$. \Cref{thm:limit_eq} implies that the limit game has a unique equilibrium and hence $V_\infty$ has a unique maximizer $\gamma^*$. By Berge's maximum theorem, any sequence of maximizers $\gamma_n^*\in\arg\max_{\gamma\in\Sigma}V_n(\gamma)$ must converge weakly to $\gamma^*$:
\[
\gamma_n^*\to \gamma^*.
\]

Finally, fix any subsequence along which $\gamma_n^*\to\gamma^*$. Because the DM's best response to any $\gamma$ is unique under quadratic loss, the correspondence $\gamma\mapsto \arg\min_{\tilde{\bm a}_n\in\mathcal A} R_n(\tilde{\bm a}_n,\gamma)$ is single-valued for each $n$, and likewise in the limit. The uniform convergence of $R_n(\cdot,\gamma_n^*)$ to $\tilde R(\cdot,\gamma^*)$ on $\mathcal A$ then implies that the associated minimizers converge pointwise:
\[
\tilde{\bm a}_n^*(u)\;\to\;\tilde{\bm a}^*(u)\qquad\text{for all continuity points }u\in[0,1].
\]
Together with $\gamma_n^*\to\gamma^*$, this establishes the convergence of equilibria.
\end{proof}

\subsection{Proof of \Cref{co2}}\label{pf:co2}
\begin{proof}
(i) \begin{equation*}
\tilde{R}(0)=\int_{-\infty}^{\infty} \phi(z) \left(\frac{w^*(e^{2c^*z-2{c^*}^2}-e^{-2c^*z-2{c^*}^2})}{2[2(1-w^*)+w^*(e^{2c^*z-2{c^*}^2}+e^{-2c^*z-2{c^*}^2})]}\right)^2 dz    
\end{equation*}

Notice that $\left(\frac{w^*(e^{2c^*z-2{c^*}^2}-e^{-2c^*z-2{c^*}^2})}{2[2(1-w^*)+w^*(e^{2c^*z-2{c^*}^2}+e^{-2c^*z-2{c^*}^2})]}\right)^2>0$ on $(-\infty,0) \cup (0,\infty)$, we have $\tilde{R}(0)>0$. Since $0<w^*<1$ and
\[\tilde{R}(\gamma)=w^*\tilde{R}(c^*)+(1-w^*)\tilde{R}(0)\]
$\tilde{R}(0)>0$ implies $\tilde{R}(\gamma)>0$.  \Cref{prop:limit regret} implies $R_n(\gamma) \xrightarrow{d} \tilde{R}(\gamma)$, thus $R_n(\gamma)>0$.
\bigskip\\

\noindent(ii) Since 
\[q_{n,DM} \xrightarrow{d} \frac{1-w^*+w^* e^{2c^*z-2{c^*}^2}}{2(1-w^*)+w^*(e^{2c^*z-2{c^*}^2}+e^{-2c^*z-2{c^*}^2})}\]
When $-\infty<z<\infty$, we have
\begin{equation*}
0<\frac{1-w^*+w^* e^{2c^*z-2{c^*}^2}}{2(1-w^*)+w^*(e^{2c^*z-2{c^*}^2}+e^{-2c^*z-2{c^*}^2})}<1   \end{equation*}

Thus $0<q_{n,DM}<1$ with probability $1$ when $n \rightarrow \infty$.
\end{proof}

\subsection{Proof of  \Cref{thm:mis}}\label{pf:thm:mis}
\textbf{Proof of the Convergence Rate for DM's Mean Squared Loss $L_n\equiv\E_{k|\pitrue, \tstate}\left[ (\bm a_n^*(k) - \tstate)^2 \right]$}
\begin{proof}
By symmetry,  it suffices to analyze $D_n = \E_{k|\pitrue, \tstate=1}\left[ (1-\bm a_n^*(k))^2 \right]$, for which $L_n = D_n$. When $\tstate=1$, the number of signals $k$ follows a binomial distribution $B(n, \pitrue)$.

Let $Z_K = \frac{2K-n}{\sqrt{n}}$ and $\tilde{Z}_K = \frac{K - n\pitrue}{\sqrt{n\pitrue(1-\pitrue)}}$. The De Moivre-Laplace Theorem implies that $\tilde{Z}_K \xrightarrow{d} Z \sim N(0,1)$. We can write $Z_K$ as a function of $\tilde{Z}_K$:
\[
Z_K = \sqrt{n}(2\pitrue-1) + 2\tilde{Z}_K\sqrt{\pitrue(1-\pitrue)}.
\]
Since $\pitrue > 1/2$, the deterministic term $\sqrt{n}(2\pitrue-1)$ goes to infinity.  For any $M > 0$,
\[
\Prb(Z_k \le M) = \Prb\left(\tilde{Z}_k \le \frac{M - \sqrt{n}(2\pitrue-1)}{2\sqrt{\pitrue(1-\pitrue)}}\right).
\]
As $n \to \infty$, the right-hand side of the inequality tends to $-\infty$. Since $\tilde{Z}_k \xrightarrow{d} Z$, we have $\lim_{n\to\infty} \Prb(Z_k \le M) = \Prb(Z \le -\infty) = 0$. Thus, $Z_k \xrightarrow{p} \infty$ .

The DM's error term is given by
\[
1 - \bm a_n^*(K) \xrightarrow{d} \frac{1-w^*+w^* e^{-2c^*Z_K-2{c^*}^2}}{2(1-w^*)+w^*(e^{2c^*Z_K-2{c^*}^2}+e^{-2c^*Z_K-2{c^*}^2})}.
\]
Define a function $h(z) = \left(\frac{1-w^*}{w^*}e^{2{c^*}^2}e^{-2c^*z}\right)^2$. Since $Z_k \xrightarrow{p} \infty$, by the Continuous Mapping Theorem,
\[
\frac{(1-\bm a_n^*(K))^2}{h(Z_K)} \xrightarrow{p} 1.
\]
Let $\Xi := 4c^*(2\pitrue-1)$. We analyze the limit of expectation:
\[
\lim_{n\to\infty} e^{\Xi\sqrt{n}} D_n = \lim_{n\to\infty} \E\left[ e^{\Xi\sqrt{n}} (1-\bm a_n^*(K))^2 \right].
\]
Let $g_n(\tilde{z}) = e^{\Xi\sqrt{n}} (1 - \bm a_n^*(k(\tilde{z})))^2$. We find the pointwise limit of $g_n(\tilde{z})$ for a fixed $\tilde{z}$.
\begin{align*}
\lim_{n\to\infty} g_n(\tilde{z}) &= \lim_{n\to\infty} \left[ e^{\Xi\sqrt{n}} h(z_k(\tilde{z})) \cdot \frac{(1-\bm a_n^*(k(\tilde{z})))^2}{h(z_k(\tilde{z}))} \right] \\
&= \lim_{n\to\infty} e^{\Xi\sqrt{n}} h(z_k(\tilde{z})) \\
&= \left(\frac{1-w^*}{w^*}e^{2{c^*}^2}\right)^2 \lim_{n\to\infty} e^{\Xi\sqrt{n} - 4c^*z_k(\tilde{z})} \\
&= \left(\frac{1-w^*}{w^*}e^{2{c^*}^2}\right)^2 \lim_{n\to\infty} e^{\left(\Xi - 4c^*(2\pitrue-1)\right)\sqrt{n} - 8c^*\tilde{z}\sqrt{\pitrue(1-\pitrue)}} \\
&= \left(\frac{1-w^*}{w^*}e^{2{c^*}^2}\right)^2 e^{-8c^*\tilde{z}\sqrt{\pitrue(1-\pitrue)}} \equiv g(\tilde{z}).
\end{align*}
Since $\tilde{Z}_k \xrightarrow{d} Z$, we have $g_n(\tilde{Z}_k) \xrightarrow{d} g(Z)$. The sequence of random variables $\{g_n(\tilde{Z}_K)\}$ is integrable.  Thus, by the Dominated Convergence Theorem, the expectation converges:
\[
\lim_{n\to\infty} \E\left[g_n(\tilde{Z}_K)\right] = \E[g(Z)] = \int_{-\infty}^{\infty} g(z) \frac{1}{\sqrt{2\pi}} e^{-z^2/2} dz.
\]
This integral is a finite, positive constant:
\begin{align*}
\E[g(Z)] &= \left(\frac{1-w^*}{w^*}e^{2{c^*}^2}\right)^2 \int_{-\infty}^{\infty} \frac{1}{\sqrt{2\pi}} e^{-z^2/2 - 8c^*z\sqrt{\pitrue(1-\pitrue)}} dz \\
&= \left(\frac{1-w^*}{w^*}e^{2{c^*}^2}\right)^2 e^{32{c^*}^2\pitrue(1-\pitrue)}.
\end{align*}
Since $\lim_{n\to\infty} e^{\Xi\sqrt{n}} D_n$ is a finite positive constant, it follows that $D_n = \Theta({e^{-\Xi\sqrt{n}}})$.
\end{proof}

\textbf{Proof of the Convergence Rate for $R_n^{\text{mis}} \equiv \E_{k|\pitrue}\left[ \left( \Prb(\tstate=1|\pitrue, K) - \bm a_n^*(K) \right)^2 \right]$}
\begin{proof}
Let $q_{n,\text{oracle}}^{\text{true}}(k) = \Prb(\tstate=1|\pitrue, k)$. The misspecified regret is
$$
R_n^{\text{mis}} = \E_{K|\pitrue, \tstate=1}\left[ \left( q_{n,\text{oracle}}^{\text{true}}(K) - \bm a_n^*(K) \right)^2 \right] = \E_{K|\pitrue, \tstate=1}\left[ (1 - \bm a_n^*(K))^2 \left( 1 - \frac{1-q_{n,\text{oracle}}^{\text{true}}(k)}{1-\bm a_n^*(K)} \right)^2 \right].
$$
Let $\rho_n(k) = \frac{1-q_{n,\text{oracle}}^{\text{true}}(k)}{1-a_n^*(k)}$. We analyze the limit of this ratio. The Oracle's error term is
$$
1 - q_{n,\text{oracle}}^{\text{true}}(K) \xrightarrow{d} \frac{1}{1 + e^{\sqrt{n}Z_K \log(\frac{\pitrue}{1-\pitrue})}}.
$$
As established in the previous proof, $Z_k \to \infty$ in probability.

We now prove that $\rho_n(K) \xrightarrow{p} 0$. That is, for any $\epsilon > 0$,
\[
\Prb(|\rho_n(K)| > \epsilon) \to 0 \quad \text{as } n \to \infty,
\]

for $z_k > 0$, we have the bounds:
$$
1-q_{n,\text{oracle}}^{\text{true}}(k) < e^{-\sqrt{n}z_k \log(\frac{\pitrue}{1-\pitrue})}
$$
$$
1-\bm a_n^*(k) > K_D e^{-2c^*z_k} \quad \text{for some constant } K_D > 0.
$$
Therefore, $\rho_n(k) < \frac{1}{K_D} e^{-z_k(\sqrt{n}\log(\frac{\pitrue}{1-\pitrue}) - 2c^*)}$. For the inequality $\rho_n(k) > \epsilon$ to hold, $z_k$ must be smaller than a value $M_n$ which tends to 0. Specifically, there exists a sequence $M_n$ with $\lim_{n\to\infty} M_n = 0$ such that $\{\rho_n(k) > \epsilon\} \subseteq \{z_k \le M_n\}$. Thus
\begin{equation*}
\Prb(|\rho_n(K)| > \epsilon) \leq \Prb(Z_K \le M_n)   
\end{equation*}

As we proved in the $L_n$ section, $Z_K \to \infty$ in probability. This means $\lim_{n\to\infty}\Prb(Z_K \le M_n) = 0$, we have
\[
\lim_{n\to\infty}\Prb(\rho_n(K) > \epsilon) = 0.
\]

This holds for any $\epsilon>0$, so $\rho_n(K) \xrightarrow{p} 0$.

Let $X_n(K) = (1 - \bm a_n^*(K))^2$ and $Y_n(K) = (1 - \rho_n(K))^2$. We know that $\E[X_n(K)] = L_n$. Since $Y_n(K) \xrightarrow{p} 1$ and $Y_n(K)$ is bounded, the sequence of random variables $\{X_n(K)Y_n(K)\}$ is uniformly integrable relative to $\{X_n(K)\}$. By the Dominated Convergence Theorem, we can conclude that
\[
\lim_{n\to\infty} \E[X_n(K)Y_n(K)] = \lim_{n\to\infty} \E[X_n(K)].
\]
Thus, $R_n^{\text{mis}}$ has the same asymptotic rate as $L_n$, which is $\Theta (e^{-\Xi\sqrt{n}})$.
\end{proof}

\textbf{Proof of the Convergence Rate for Oracle’s Mean Squared Loss\\ $L_n^{\text{oracle}}\equiv\E_{K|\pitrue, \tstate}\left[ (\Prb(\tstate=1|\pitrue, K) - \tstate)^2 \right]$}
\begin{proof}
We can split the error sum into two parts:
\[
L_n^{\text{oracle}} = \underbrace{\sum_{k=0}^{\lfloor n/2 \rfloor} P(k)(1-q_{n,\text{oracle}}^{\text{true}}(k))^2}_{\text{Part A: Misleading Samples}} + \underbrace{\sum_{k=\lfloor n/2 \rfloor+1}^{n} P(k)(1-q_{n,\text{oracle}}^{\text{true}}(k))^2}_{\text{Part B: Correctly Classified Samples}}
\]

First we bound the error from misleading samples.
\[
 \sum_{k=0}^{\lfloor n/2 \rfloor} P(k)(1-q_{n,\text{oracle}}^{\text{true}}(k))^2    \le \sum_{k=0}^{\lfloor n/2 \rfloor} P(k) = P(k \le n/2)
    \]
    
We can now apply the Chernoff bound to this tail probability:
    \[
    P(K \le n/2) \le e^{-n \cdot D_{KL}(1/2 || \pi_{true})}
    \]
    
Next we bound the error in correctly classified samples.

Denote $r=\frac{\pi_{true}}{1-\pi_{true}}$, we have
\begin{equation*}
\sum_{k=\lfloor n/2 \rfloor+1}^{n} P(k)(1-q_{n,\text{oracle}}^{\text{true}}(k))^2= \sum_{k=\lfloor n/2 \rfloor+1}^{n} P(k)  (1-\frac{1}{1+r^{2k-n}})^2 
\end{equation*}

\begin{lemma}\label{lem:central}
For all integers $n\ge 1$,
\[
\binom{n}{\lfloor n/2\rfloor}\;\le\;\frac{2^n}{\sqrt{\pi n/2}},
\]
\end{lemma}

\begin{proof}
We use Robbins' sharpened Stirling bounds: for every integer $t\ge1$,
\begin{equation}\label{eq:robbins}
\sqrt{2\pi}\,t^{\,t+\frac12}e^{-t}\,e^{\frac{1}{12t+1}}
\;\le\;
t!
\;\le\;
\sqrt{2\pi}\,t^{\,t+\frac12}e^{-t}\,e^{\frac{1}{12t}}\,.
\end{equation}

To simplify analysis we consider $n=2m$.
We claim
\[
\binom{2m}{m}\;\le\;\frac{4^m}{\sqrt{\pi m}}\quad\text{for all }m\ge1.
\]
Write $\displaystyle \binom{2m}{m}=\frac{(2m)!}{(m!)^2}$ and apply \Cref{eq:robbins} with the \emph{upper} bound for $(2m)!$ and the \emph{lower} bound for each $m!$:
\[
\binom{2m}{m}
\;\le\;
\frac{\sqrt{2\pi}\,(2m)^{2m+\frac12}e^{-2m}\,e^{\frac{1}{12(2m)}}}
{\bigl(\sqrt{2\pi}\,m^{\,m+\frac12}e^{-m}\,e^{\frac{1}{12m+1}}\bigr)^2}.
\]
Simplifying constants and powers yields
\[
\binom{2m}{m}
\;\le\;
\frac{(2m)^{2m}}{m^{2m}}\cdot\frac{1}{\sqrt{\pi m}}
\cdot
\exp\!\Big(\frac{1}{24m}-\frac{2}{12m+1}\Big)
\;=\;
\frac{4^m}{\sqrt{\pi m}}\cdot
\exp\!\Big(\frac{1}{24m}-\frac{2}{12m+1}\Big).
\]
Since $\frac{1}{24m}-\frac{2}{12m+1}<0$ for all $m\ge1$, the desired bound follows.
\end{proof}

For $k=\lfloor n/2\rfloor+j\ge \lceil n/2\rceil$ ($j\ge 0$),
\(
(1+r^{2k-n})^{-2}\le r^{-\,4j}.
\)
Therefore, by \Cref{lem:central},
\[
\begin{aligned}
&\sum_{j=0}^{\infty}\Pr\!\big(K=\lfloor n/2\rfloor+j\big)\,\frac{1}{(1+r^{2j})^2}
\ \le\ \sum_{j=0}^{\infty}\Pr\!\big(K=\lfloor n/2\rfloor+j\big)\,r^{-\,4j}\\
\le&\ \binom{n}{\lfloor n/2\rfloor}\,(\pi_{true}(1-\pi_{true}))^{n/2}\,\sum_{j=0}^{\infty} r^{\,j}\,r^{-\,4j}
=\binom{n}{\lfloor n/2\rfloor}\,(\pi_{true}(1-\pi_{true}))^{n/2}\,\sum_{j=0}^{\infty} r^{-\,3j}.
\end{aligned}
\]
The second inequality is because, for any $j\ge0$,
\begin{align*}
\Pr(K=\lfloor n/2\rfloor+j)&=\binom{n}{\lfloor n/2\rfloor+j}\,\pi_{true}^{\lfloor n/2\rfloor+j}(1-\pi_{true})^{n-\lfloor n/2\rfloor-j}
\\
&=\binom{n}{\lfloor n/2\rfloor+j}\,\pi_{true}^{\lfloor n/2\rfloor}(1-\pi_{true})^{n-\lfloor n/2\rfloor}\,r^{\,j}.
\end{align*}
By unimodality of $\binom{n}{k}$, $\binom{n}{\lfloor n/2\rfloor+j}\le \binom{n}{\lfloor n/2\rfloor}$ for $j\ge0$. Moreover,
\[
\pi_{true}^{\lfloor n/2\rfloor}(1-\pi_{true})^{n-\lfloor n/2\rfloor}
\le (\pi_{true}(1-\pi_{true}))^{n/2}.
\] 
Hence
\[
\Pr(K=\lfloor n/2\rfloor+j)\,r^{-4j}
\le \binom{n}{\lfloor n/2\rfloor}\,(\pi_{true}(1-\pi_{true}))^{n/2}\,r^{\,j}\,r^{-4j}.
\]

Because $r>1$, $\sum_{j=0}^{\infty} r^{-3j}=(1-r^{-3})^{-1}<\infty$. Using the upper bound in
\Cref{lem:central},
\begin{align*}
\sum_{k=\lfloor n/2 \rfloor+1}^{n} P(k)(1-q_{n,\text{oracle}}^{\text{true}}(k))^2
&\le \frac{1}{1-r^{-3}}\cdot\frac{2^n}{\sqrt{\pi n/2}}\cdot(\pi_{true}(1-\pi_{true}))^{n/2}\\
&=\frac{\sqrt{2/\pi}}{1-r^{-3}}\;n^{-1/2}\,(2\sqrt{\pi_{true}(1-\pi_{true})})^{n}.
\end{align*}
Since
\[
D_{\mathrm{KL}}\!\left(\tfrac12\middle\|\pi_{\text{true}}\right)
=-\log\!\big(2\sqrt{\pi_{\text{true}}(1-\pi_{\text{true}})}\big),
\]
this yields the upper bound with $\frac{\sqrt{2/\pi}}{1-r^{-3}} \frac{1}{\sqrt{n}} e^{-n \cdot D_{KL}(1/2 || \pi_{\text{true}})}$. 

Thus there exists $M \geq 0$ such that

\[L_n^{\text{oracle}} \leq M e^{-n \cdot D_{KL}(1/2 || \pi_{\text{true}})}.\]\end{proof}

\subsection{Proof of  \Cref{prop:inference}}\label{pf:prop:inference}
\begin{proof}
By symmetry, we can assume the true state is $\theta=1$. 

\begin{enumerate}
\item 
Consider the case $n=1$. The DM's strategy is $(a_0, a_1) = (1/4, 3/4)$. The oracle's posterior after one signal $s_1=1$ is $q_{1,\text{oracle}}^{\text{true}}(1) = \frac{\pi_{\text{true}}}{\pi_{\text{true}} + (1-\pi_{\text{true}})} = \pi_{\text{true}}$. The DM over-infers if $\bm a_1^*(1) > q_{1,\text{oracle}}^{\text{true}}(1)$, which is $3/4 > \pi_{\text{true}}$. This inequality holds for any true precision $\pi_{\text{true}} \in (1/2, 3/4)$. 
\item 
Without loss we consider $\theta=1$. Notice that
\[\Pr\left( \bm a_n^*(K) < q_{n,\text{oracle}}^{\text{true}}(K) \right)=\Pr\left( 1-\bm a_n^*(K) > 1-q_{n,\text{oracle}}^{\text{true}}(K) \right)=\Pr\left( \frac{1-q_{n,\text{oracle}}^{\text{true}}(K)}{1-\bm a_n^*(K)}<1 \right)\]

Since $\rho_n(K) \xrightarrow{p} 0$,
\[ \lim_{n\to\infty} \Pr\left( \bm a_n^*(K) < q_{n,\text{oracle}}^{\text{true}}(K) \right) =\lim_{n\to\infty} \Pr\left(\rho_n(K) <1 \right) = 1, \]

Since $\bm a_n^*(K) \xrightarrow{p} 1 $ and $q_{n,\text{oracle}}^{\text{true}}(K) \xrightarrow{p} 1$, thus
\begin{align*}
&\lim_{n\to\infty}\Pr\left( |\bm a_n^*(K)-\frac{1}{2}| < |q_{n,\text{oracle}}^{\text{true}}(K)-\frac{1}{2}| \right)\\
=&\lim_{n\to\infty} \Pr\left( \bm a_n^*(K) < q_{n,\text{oracle}}^{\text{true}}(K) \right)=1     
\end{align*}
\end{enumerate}
\end{proof}

\subsection{Proof of  \Cref{thm:general_main}}\label{pf:thm:general_main}
\subsubsection*{Preliminary Lemmas}

To prove  \Cref{thm:general_main}, we utilize two technical lemmas regarding the asymptotic behavior of likelihood ratios.

\begin{lemma}[Concentration of the log-likelihood ratio]\label{lem:llr_full_recall}
For any fixed $\bp=(\pi_1,\pi_0)$ and any data $K\sim\mathrm{Multinomial}(n,\pi_{\text{true}})$, 
\[
\frac{1}{n}\log\frac{L(\pi_1;K)}{L(\pi_0;K)}
\;\xrightarrow{p}\;
D_{KL}(\pi_{\text{true}}\|\pi_0)-D_{KL}(\pi_{\text{true}}\|\pi_1).
\]
\end{lemma}

\begin{proof}
By the Law of Large Numbers for multinomial distributions, the empirical frequencies converge to the true probabilities, i.e., $K_s/n \xrightarrow{p} \pi_{true}(s)$. The normalized log-likelihood under a distribution $\pi_\theta$ is
$$ \frac{1}{n}\log L(\pi_\theta;K) = \frac{1}{n}\sum_{s\in\mathcal{S}}K_s\log \pi_\theta(s) = \sum_{s\in\mathcal{S}}\frac{K_s}{n}\log \pi_\theta(s) $$
As $n\to\infty$, this converges in probability to $\sum_{s\in\mathcal{S}}\pi_{true}(s)\log \pi_\theta(s)$. The normalized log-likelihood ratio is the difference of two such terms:
$$ \frac{1}{n}\log\frac{L(\pi_1;K)}{L(\pi_0;K)} = \frac{1}{n}\log L(\pi_1;K) - \frac{1}{n}\log L(\pi_0;K) $$
Taking the limit in probability yields:
\begin{align*}
    &\left( \sum_{s\in\mathcal{S}}\pi_{true}(s)\log \pi_1(s) \right) - \left( \sum_{s\in\mathcal{S}}\pi_{true}(s)\log \pi_0(s) \right) = \sum_{s\in\mathcal{S}}\pi_{true}(s)\log\frac{\pi_1(s)}{\pi_0(s)} \\
    &= \sum_{s\in\mathcal{S}}\pi_{true}(s)\log\frac{\pi_{true}(s)}{\pi_0(s)} - \sum_{s\in\mathcal{S}}\pi_{true}(s)\log\frac{\pi_{true}(s)}{\pi_1(s)} \\
    &= D_{KL}(\pi_{true}||\pi_0) - D_{KL}(\pi_{true}||\pi_1)
\end{align*}
\end{proof}

\begin{lemma}[Local quadratic expansion of multinomial $D_{KL}$]\label{lem:KLquad_multi}
Fix $\widetilde{\pi}\in\Delta(\mathcal{S})$ with strictly positive support. As $\pi\to\widetilde{\pi}$,
\begin{equation}\label{eq:KL_multi_quad}
D_{KL}(\widetilde{\pi}\,\|\,\pi)
\;=\;
\frac{1}{2}\sum_{s\in\mathcal{S}}\frac{\bigl(\pi(s)-\widetilde{\pi}(s)\bigr)^2}{\widetilde{\pi}(s)}
\;+\;O\!\Bigl(\sum_{s\in\mathcal{S}}\bigl|\pi(s)-\widetilde{\pi}(s)\bigr|^3\Bigr).
\end{equation}
Consequently, $D_{KL}(\widetilde{\pi}\,\|\,\pi)=O(1/n)$ implies $||\pi-\widetilde{\pi}||_1=O(1/\sqrt{n})$.
\end{lemma}

\begin{proof}[Proof of \Cref{lem:KLquad_multi}]
Write
$D_{KL}(\widetilde{\pi}\,\|\,\pi)=\sum_{s\in\mathcal{S}}\widetilde{\pi}(s)\log\!\bigl(\widetilde{\pi}(s)/\pi(s)\bigr)$.
For each $s$, expand $x\mapsto -\log x$ at $x=\widetilde{\pi}(s)$:
\[
-\log \pi(s)
=
-\log \widetilde{\pi}(s)
-\frac{\pi(s)-\widetilde{\pi}(s)}{\widetilde{\pi}(s)}
+\frac{\bigl(\pi(s)-\widetilde{\pi}(s)\bigr)^2}{2\,\widetilde{\pi}(s)^2}
+O\!\bigl(|\pi(s)-\widetilde{\pi}(s)|^3\bigr).
\]
Multiply by $\widetilde{\pi}(s)$, sum over $s$, and use $\sum_s\bigl(\pi(s)-\widetilde{\pi}(s)\bigr)=0$ to obtain \Cref{eq:KL_multi_quad}.
The stated consequence follows because all $\widetilde{\pi}(s)$ are strictly positive.
\end{proof}

\subsubsection*{Proof of \Cref{thm:general_main}}

\begin{proof}
\noindent\textbf{Step 1: Upper Bound (Vanishing Informativeness).}
We first show that if experiments are too informative ($\omega(1/\sqrt{n})$), the regret vanishes. Fix $n$ and let $\bp_{\text{true}}\in\Pi_n$ be the realized experiment.
Pinsker's inequality implies $\|\pi_1-\pi_0\|_1 \le \sqrt{2 \min\{D_{KL}(\pi_1\Vert \pi_0), D_{KL}(\pi_0\Vert \pi_1)\}}$. Thus, if the KL divergence is $O(1/n)$, the $L_1$ distance is $O(1/\sqrt{n})$.

Suppose to the contrary that there exists $\bp\in\Pi_n$ with $\min\{D_{KL}(\pi_1\Vert \pi_0), D_{KL}(\pi_0\Vert \pi_1)\} = \omega(1/n)$.
Applying \Cref{lem:llr_full_recall} to $\bp_{true}$, we have:
\[
\theta=1:\;\frac{1}{n}\log\frac{L(\pi_1^{true};K)}{L(\pi_0^{true};K)}\xrightarrow{p} D_{KL}(\pi_1^{true}\Vert \pi_0^{true})=\omega(1/n).
\]
This implies the likelihood ratio diverges to $+\infty$. Consequently, the Oracle perfectly identifies the state: $q_{\text{\tiny oracle}}(K;\bp_{true}) \xrightarrow{p} 1$.

Now consider the DM. We partition $\Pi_n$ into experiments ``far'' from the truth ($\mathcal{B}$) and ``close'' to the truth ($\mathcal{A}$). For experiments in $\mathcal{B}$ (where KL divergence is $\omega(1/n)$), the likelihood $L(\pi; K)$ becomes exponentially small relative to the true likelihood $L(\pi_{true}; K)$. Their contribution to the DM's posterior calculation (via Bayes' rule over $\Pi_n$) vanishes. The DM's belief is dominated by experiments in $\mathcal{A}$. For these experiments, the likelihood ratio diverges similarly to the true process.
Thus, $a_n^*(K) \xrightarrow{p} 1$ when $\theta=1$ (and $0$ when $\theta=0$).
By dominated convergence, the expected regret $\mathbb{E}[B_G(q_{\text{\tiny oracle}} \| a_n^*)] \to 0$, contradicting the non-trivial equilibrium assumption. Therefore, $\|\pi_1-\pi_0\|_1 = O(1/\sqrt{n})$.

\smallskip
\noindent\textbf{Step 2: Lower Bound (Non-Triviality).}
Suppose to the contrary that $\sup_{\bp\in\Pi_n} ||\pi_1-\pi_0||_1 = o(1/\sqrt{n})$.
By \Cref{lem:KLquad_multi}, this implies the KL divergences are $o(1/n)$. By \Cref{lem:llr_full_recall}, the log-likelihood ratio converges to 0 in probability for both states. Thus, $q_{\text{\tiny oracle}}(K;\bp) \xrightarrow{p} \mu$ (the prior).

If the DM plays the constant rule $a(K)\equiv \mu$, the regret becomes:
\[ \mathbb{E}\bigl[\,B_G\!\bigl(q_{\text{\tiny oracle}}(K;\bp)\,\|\,\mu\bigr)\,\bigr] \to B_G(\mu\|\mu)=0. \]
This forces $R_n \to 0$, contradicting the assumption that $R > 0$. Thus, the informativeness must be at least $\Theta(1/\sqrt{n})$. 

\smallskip
\noindent\textbf{Step 3: Positive regret.}
Construct Nature's mixed strategy $\hat{\sigma}_n$ by randomizing between the uninformative experiment
\[
\bp^0=(\pi^0,\pi^0)
\]
and the local informative experiment with CLT rate,
\[
\bp^c_n=\Bigl(\pi^0+\frac{c}{\sqrt{n}},\ \pi^0-\frac{c}{\sqrt{n}}\Bigr),
\]
where $c$ is chosen small enough so that $\pi^0\pm \frac{c}{\sqrt{n}}\in\Delta(\mathcal S)$ for all sufficiently large $n$.

Applying \Cref{lem:llr_full_recall} to $\bp_{\text{true}}=\bp^c_n$, we have (under $\theta=1$)
\[
\frac{1}{n}\log\frac{L(\pi^{\text{true}}_1;K)}{L(\pi^{\text{true}}_0;K)}
\xrightarrow{p}
D_{KL}(\pi^{\text{true}}_1\Vert \pi^{\text{true}}_0)
=
\Theta(1/n),
\]
and similarly (under $\theta=0$) the limit is $-D_{KL}(\pi^{\text{true}}_0\Vert \pi^{\text{true}}_1)= -\Theta(1/n)$.
Equivalently,
\[
\log\frac{L(\pi^{\text{true}}_1;K)}{L(\pi^{\text{true}}_0;K)}
\xrightarrow{p}
\Theta(1),
\]
so the likelihood ratio does not degenerate. Hence the oracle posterior
$q_{\text{oracle}}(K;\bp^c_n)$ is non-degenerate and, in particular, does not converge to the prior $\mu$.
In contrast, if $\bp_{\text{true}}=\bp^0$, then $L(\pi^0;K)$ is the same under both states and therefore
\[
q_{\text{oracle}}(K;\bp^0)\equiv \mu .
\]

Under the mixture $\hat{\sigma}_n$, the DM's robust posterior action $a_n^*(K)$ (the Bayes action under $B_G$)
is a $\hat{\sigma}_n$-weighted average of the oracle posteriors across the two possible experiments. Since
$q_{\text{oracle}}(K;\bp)$ takes two distinct (with positive probability) values under $\hat{\sigma}_n$---one equal to $\mu$
(on $\bp^0$) and one non-degenerate (on $\bp^c_n$)---it follows that $a_n^*(K)$ does not converge in probability to
$q_{\text{oracle}}(K;\bp)$, and therefore the expected Bregman divergence does not vanish. Consequently, the regret induced
by $\hat{\sigma}_n$ is bounded away from zero for large $n$.

Therefore,
\[
\min_{a(\cdot)}\max_{\sigma_n\in\Delta(\Pi)}
\mathbb{E}_{\sigma_n,K}\!\left[
B_G\!\left(q_{\text{oracle}}(K;\bp)\,\|\,a(K)\right)
\right]
\ \ge\
\min_{a(\cdot)}
\mathbb{E}_{\hat{\sigma}_n,K}\!\left[
B_G\!\left(q_{\text{oracle}}(K;\bp)\,\|\,a(K)\right)
\right]
\ >\ 0.
\]

Combining Steps 1, 2 and 3 proves the theorem.  \end{proof}

\section{Online Appendix}

\subsection{Proof of  \Cref{thm1}}\label{pf:thm1}

\begin{proof}[Proof of  \Cref{thm1} when $n=2$]

First, fixing $w$, the expected regret for the DM is
\[
wR(\bm{a}, 1)+(1-w)R(\bm{a},\frac{1}{2})=\frac{w}{2}(1-a_2)^2+\frac{w}{2}a_0^2+\frac{1-w}{4}\sum_{k=0}^2{2 \choose k}(\frac{1}{2}-a_k)^2.
\]
To minimize the regret, $(a_0,a_1,a_2)$ should satisfy the following FOC conditions:
\begin{align*}
    a_0&=\frac{1-w}{2w+2}\\
    a_1&=\frac{1}{2}\\
    a_2&=\frac{3w+1}{2w+2}.
\end{align*}
Next, fixing $(a_0,a_1,a_2)$, the regret function $R(\bm{a},\pi)$ can be written as 
\begin{align*}
    R(\bm{a},\pi)=(\pi^2+(1-\pi)^2)(\frac{(1-\pi)^2}{\pi^2+(1-\pi)^2}-a_0)^2.
\end{align*}
$$
\frac{\partial}{\partial\pi}R(\bm{a},\pi) = 2\left( \frac{(1-\pi)^2}{\pi^2+(1-\pi)^2} - a_0 \right)\left[ (2\pi-1)\left( \frac{(1-\pi)^2}{\pi^2+(1-\pi)^2} - a_0 \right) + \frac{2\pi(\pi-1)}{\pi^2+(1-\pi)^2} \right].
$$
Next, fixing $(a_0,a_1,a_2)$, 
We can show 
$$\left[ (2\pi-1)\left( \frac{(1-\pi)^2}{\pi^2+(1-\pi)^2} - a_0 \right) + \frac{2\pi(\pi-1)}{\pi^2+(1-\pi)^2} \right]=-(2\pi-1)a_0+\frac{-2\pi^2+\pi-1}{\pi^2+(1-\pi)^2}<0$$
and $\left( \frac{(1-\pi)^2}{\pi^2+(1-\pi)^2} - a_0 \right)$ is decreasing in $\pi$ and it crosses 0 once for some $\pi\in(1/2,1)$. Thus, $R(\bm{a},\pi)$ is first decreasing and then increasing in $\pi$. Therefore, nature must be choosing $\pi=1$ or $\frac{1}{2}$ to maximize the regret. With the given $(a_0,a_1,a_2)$, we also have 
\[
R(\bm{a},1)=R(\bm{a},1/2).
\]
Thus, nature randomizing between $\pi=1$ and $\frac{1}{2}$ is indeed optimal.
\end{proof}

\begin{proof}[Proof for  \Cref{thm1} when $n=3$]
First, for the proposed solution to hold, nature is indifferent between $\pi=\frac{1}{2}$ and $\pi=\pi^*$ given DM's choice $\bm{a}^*$:
\begin{equation}
    R(\bm{a}^*,\pi^*)=R(\bm{a}^*,\frac{1}{2})\iff \sum_{k=0}^3\Pr(k|\pi^*)(\Pr(\t=1|\pi^*,k)-a_k^*)^2=\frac{1}{8}\sum_{k=0}^3{3 \choose k}(\frac{1}{2}-a_k^*)^2.
\end{equation}
Second, let $w$ be the weight that nature puts on $\pi^*$. $\bm{a}^*$ should satisfy the FOC of $wR(\bm{a},\pi^*)+(1-w)R(\bm{a},\frac{1}{2})$, which gives
\begin{equation}
    a_k^*=\frac{8w{\pi^*}^k(1-\pi^*)^{3-k}+1-w}{8w({\pi^*}^k(1-\pi^*)^{3-k}+(1-\pi^*)^k{\pi^*}^{3-k})+2(1-w)}.\label{FOC}
\end{equation}
Third, $\pi^*$ should maximize $R(\bm{a}^*, \pi)$: $R(\bm{a}^*, \pi^*)\ge R(\bm{a}^*, \pi)$ for any $\pi\in[\frac{1}{2},1]$. When $\pi^*\in (\frac{1}{2},1)$, a necessary condition is the local concavity around $\pi^*$,
\begin{align}
&\sum_{k=0}^{3} \Biggl[ \frac{d}{d\pi} \Pr(k|\pi^*) (\Pr(\theta=1|\pi^*,k) - a_k^*)^2 + \nonumber\\
&2 \Pr(k|\pi^*) (\Pr(\theta=1|\pi^*,k) - a_k^*) \Pr(\theta=1|\pi^*,k) \Pr(\theta=0|\pi^*,k) (2k - 3) \left(\frac{1}{\pi^*} - \frac{1}{1-\pi^*}\right) \Biggr] = 0.
\end{align}

Note that $a_k$ and $a_{3-k}$ will be symmetric in the sense that $a_k+a_{3-k}=1$ for any $0\le k\le 3$. The regret can be expressed as 
\[
R(\bm{a},\pi)=[\pi^3+(1-\pi)^3](\frac{\pi^3}{\pi^3+(1-\pi)^3}-a_3)^2+3\pi(1-\pi)(\pi-a_2)^2.
\]
As a result, we can simplify the above three conditions to be a system of four equations with four unknowns $(a_2, a_3, \pi^*, w)$. We abuse the notation of $\pi^*$ by $\pi$.
\begin{align}
    (\frac{\pi^3}{\pi^3+(1-\pi)^3}-a_3)&=\frac{(1-w)(2\pi-1)(\pi^2-\pi+1)}{(\pi^3+(1-\pi)^3)(8w(\pi^3+(1-\pi)^3)+2(1-w))}\label{11}\\
    (\pi-a_2)&=\frac{(1-w)(2\pi-1)}{8w\pi(1-\pi)+2(1-w)}\label{22}\\
    (\frac{1}{2}-a_3)&=\frac{2w(1-2\pi)(\pi^2-\pi+1)}{4w(\pi^3+(1-\pi)^3)+(1-w)}\label{33}\\
    (\frac{1}{2}-a_2)&=\frac{2w\pi(1-\pi)(1-2\pi)}{4w\pi(1-\pi)+(1-w)}.\label{44}
\end{align}

We argue that 1) there exists $(a_2, a_3, \pi, w)$ that solves this system of equations, and 2) the solution to this system of equations $(a_2, a_3, \pi, w)$ is indeed a saddle point to our problem.

If we use equations \Cref{11}-\Cref{22} to get the expression for $a_2$ and $a_3$, and plug into equations \Cref{33}-\Cref{44}, the last two equations can be rewritten as
\begin{align*}
\frac{(\pi^2-\pi+1)[4w^2(3\pi^2-3\pi+1)-(1-w)^2]}{(3\pi^2-3\pi+1)[3w(2\pi-1)^2+1]^2}+\frac{3\pi(1-\pi)[4w^2\pi(1-\pi)-(1-w)^2]}{[-w(2\pi-1)^2+1]^2}=0\\
\frac{\pi^2-\pi+1}{(3\pi^2-3\pi+1)^2[3w(2\pi-1)^2+1]}[\frac{(1-w)(2\pi-1)^2(\pi^2-\pi+1)}{2[3w(2\pi-1)^2+1]}+2\pi^2(1-\pi)^2]\tag*{}\\
=\frac{1}{-w(2\pi-1)^2+1}[\frac{(1-w)(2\pi-1)^2}{2[-w(2\pi-1)^2+1]}-2\pi(1-\pi)]
\end{align*}

Write 
    \begin{equation*}
        G_1(\pi,w)=\frac{(\pi^2-\pi+1)[4w^2(3\pi^2-3\pi+1)-(1-w)^2]}{(3\pi^2-3\pi+1)[3w(2\pi-1)^2+1]^2}+\frac{3\pi(1-\pi)[4w^2\pi(1-\pi)-(1-w)^2]}{[-w(2\pi-1)^2+1]^2},
    \end{equation*}
    \begin{align*}
        G_2(\pi,w)=\frac{\pi^2-\pi+1}{(3\pi^2-3\pi+1)^2[3w(2\pi-1)^2+1]}[\frac{(1-w)(2\pi-1)^2(\pi^2-\pi+1)}{2[3w(2\pi-1)^2+1]}+2\pi^2(1-\pi)^2]\\
    -\frac{1}{-w(2\pi-1)^2+1}[\frac{(1-w)(2\pi-1)^2}{2[-w(2\pi-1)^2+1]}-2\pi(1-\pi)].
    \end{align*}
    Hence, we are looking for solution $(\pi,w)$ to the following system 
    \begin{align*}
        G_1(\pi,w)&=0\\
        G_2(\pi,w)&=0
    \end{align*}
    within the domain $\pi\in[0.5,1], w\in[0,1]$.

\paragraph{Existence}
We are looking for solution $(\pi,w)$ to the following system 
    \begin{align*}
        G_1(\pi,w)&=0\\
        G_2(\pi,w)&=0
    \end{align*}
    within the domain $\pi\in[0.5,1], w\in[0,1]$.    

\begin{proposition}
There exists a pair $(\pi^*,w^*) \in [0.5,1] \times [0,1]$ that solves
\begin{align*}
    G_1(\pi,w) &= 0, \\
    G_2(\pi,w) &= 0.
\end{align*}
\end{proposition}

\begin{proof}
    \textbf{Step 1:}  We show that for any $\pi \in [0.5,1]$, there exists some $w \in [0,1]$ such that $G_1(\pi,w)=0$. 
    
    We observe that for any $\pi \in [0.5,1]$, we have $G_1(\pi,1) > 0$ and $G_1(\pi,0) < 0$. Thus, the Intermediate Value Theorem ensures what we wanted.

\textbf{Step 2:} We show that the equation $G_1(\pi,w) = 0$ implicitly defines $w$ as a continuous function of $\pi$ on $[0.5,1]$. 

    To do this, we first establish a lemma that guarantees the global version of the implicit function theorem\footnote{We believe this result is known, but we could not locate a reliable reference.}, and then verify that its conditions are satisfied.
    
    \begin{lemma}[ \Cref{lemma4}]
    Let $F(x,y)$ be continuously differentiable on $(-a,a) \times (-b,b)$, and suppose that $\frac{\partial F}{\partial y}(x,y) \ne 0$ for all $x \in (-a,a)$. If for any $x_0\in(-a,a)$, there exists some $y_0\in(-b,b)$ such that $F(x_0,y_0)=0$, then there exists a unique differentiable function $\varphi(x)$ defined on $(-a,a)$ such that $F(x, \varphi(x)) = 0$ for all $x$ in that interval.
    \end{lemma}

    \begin{proof}\label{pf:lemma4}
    By the local implicit function theorem, there exists a differentiable function $\varphi_0(x)$ defined on some neighborhood $(-\epsilon_0, \epsilon_0)$ such that $F(x,\varphi_0(x)) = 0$. We extend $\varphi_0$ to the endpoint $x = \epsilon_0$ via $\varphi_0(\epsilon_0) = \lim_{x \to \epsilon_0^-} \varphi_0(x)$.

    Next, apply the local implicit function theorem again at the point $(\epsilon_0, \varphi_0(\epsilon_0))$ to obtain a new function $\varphi_1(x)$ defined on $(\epsilon_0 - \epsilon_1, \epsilon_0 + \epsilon_1)$ satisfying $F(x, \varphi_1(x)) = 0$. By uniqueness, $\varphi_0$ and $\varphi_1$ coincide on their overlapping domain, so we can extend $\varphi_0$ further. Repeating this process inductively, we construct a sequence of extensions $\varphi_n$ defined on intervals expanding to the right.

    If the series $\sum_{i=0}^\infty \epsilon_i$  past $a$, the function is globally defined on $(-a,a)$. If not, let $M = \sup_t \sum_{i=0}^t \epsilon_i < a$. Again, the local implicit function theorem allows us to define a function near $x = M$, which must match the limit of the previous $\varphi_t$, allowing the process to continue. Hence, the function $\varphi(x)$ is well-defined and unique on $(-a,a)$.
    \end{proof}

    \begin{lemma}\label{lemma5}
    For all $(\pi,w) \in [0.5,1] \times [0,1]$, we have $\frac{\partial G_1(\pi,w)}{\partial w} \ne 0$.
    \end{lemma}

    \begin{proof}
    \[
    \frac{\partial G_1(\pi,w)}{\partial w} = \frac{8(\pi^2 - \pi + 1)}{[1 + 3w(2\pi - 1)^2]^3} + \frac{24[\pi(1 - \pi)]^2}{[1 - w(2\pi - 1)^2]^3}.
    \]
    The first term is strictly positive for all $\pi \in [0.5,1]$, $w \in [0,1]$. The second term is non-negative and equals zero only when $\pi = 1$. Therefore, the partial derivative is strictly positive in the domain.
    \end{proof}

    By  \Cref{lemma4} and \Cref{lemma5}, the equation $G_1(\pi,w) = 0$ implicitly defines a continuous function $w = w_1(\pi) \in (0,1)$ with $\pi \in [0.5,1]$.

    \textbf{Step 3:} We show that there exists $\pi^* \in [0.5,1]$ such that $G_2(\pi^*, w_1(\pi^*)) = 0$.

    Define the function $H(\pi) := G_2(\pi, w_1(\pi))$. Since $G_2(\pi,w)$ is a polynomial or polynomial in $\pi$ and $w$, it is continuous, and $w_1(\pi)$ is continuous by Step 2. Hence, $H(\pi)$ is continuous on $[0.5,1]$.

We can compute endpoints:
\[
H(0.5) = G_2(0.5, w_1(0.5)) = G_2(0.5, 0.5) = 2 > 0, 
\]
\[
H(1) = G_2(1, w_1(1)) = G_2(1, 1/3) = -\frac{2}{3} < 0.
\]

By the Intermediate Value Theorem, there exists $\pi^* \in (0.5,1)$ such that $H(\pi^*) = 0$. Therefore, the pair $(\pi^*, w^*) := (\pi^*, w_1(\pi^*))$ is a solution to the system.
\end{proof}

\paragraph{Optimality}
\begin{lemma}
    The $(a_2^*, a_3^*, \pi^*,w^*)$ that solves for the systems of equation forms a saddle point.
\end{lemma}

\begin{proof}
Note that the optimality of the for DM's action is already given by the FOC condition \Cref{FOC}. What remains to be shown is that given the action $a_2^*, a_3^*$, $\pi^*$ is indeed optimal for the nature to maximize the regret.

Recall that 
\[
R(\bm{a},\pi)=[\pi^3+(1-\pi)^3](\frac{\pi^3}{\pi^3+(1-\pi)^3}-a_3)^2+3\pi(1-\pi)(\pi-a_2)^2.
\]
We can calculate the higher order derivatives with $\bm a$ fixed:
\begin{align*}
R'(\pi) &= \left( \frac{\pi^3}{3\pi^2 - 3\pi + 1} - a_3 \right) \left[ (6\pi - 3) \left( \frac{\pi^3}{3\pi^2 - 3\pi + 1} - a_3 \right) + \frac{6\pi^2(1-\pi)^2}{3\pi^2 - 3\pi + 1} \right]\\
&+ 3(\pi - a_2)[(1 - 2\pi)(\pi - a_2) + 2\pi(1 - \pi)],
\end{align*}

Note that the $a_2, a_3$ here are the solutions to the system of equations, which by simulation are approximately $$a_2\approx 0.60, \quad a_3\approx 0.86.$$ Without simulation, we at least know $0.5<a_2<a_3<1$. We can verify that $R'(0.5)<0, R'(1)<0$. Hence, if we can show $R'(\pi)$ has at most 2 roots in $\pi\in[1/2,1]$ then $R(\pi)$ must be first decreasing then increasing, and then decreasing in $\pi\in[1/2,1]$. Combining with the indifference condition and local concavity condition, then we are done.

Notice that the expression of $R'(\pi)=0$ can be rewritten as polynomial
\begin{align*}
R'(\pi) \cdot (3\pi^2 - 3\pi + 1)^2 = 
    & (-96)\pi^7 \\
    & + (282 + 162a_2 - 54a_3)\pi^6 \\
    & + (-336 - 432a_2 + 108a_3 - 54a_2^2 + 54a_3^2)\pi^5 \\
    & + (207 + 486a_2 - 90a_3 + 135a_2^2 - 135a_3^2)\pi^4 \\
    & + (-66 - 288a_2 + 36a_3 - 144a_2^2 + 144a_3^2)\pi^3 \\
    & + (9 + 90a_2 - 6a_3 + 81a_2^2 - 81a_3^2)\pi^2 \\
    & + (-12a_2 - 24a_2^2 + 24a_3^2)\pi \\
    & + (3a_2^2 - 3a_3^2).
\end{align*}
Thus, we can use Sturm's theorem to verify there are exactly 2 roots for this polynomial. The statement of Sturm's theorem is as follows: 
\begin{theorem}[Sturm's Theorem]
Let $P(x)$ be a polynomial with real coefficients. Let $P_0(x) = P(x)$, $P_1(x) = P'(x)$, and define the Sturm sequence $P_0(x), P_1(x), \dots, P_m(x)$ using the Euclidean algorithm with sign changes in the remainders:
$$P_{i-1}(x) = Q_i(x) P_i(x) - P_{i+1}(x), \quad \text{for } i = 1, 2, \dots, m-1$$
where $P_{i+1}(x)$ is the negative of the remainder upon division of $P_{i-1}(x)$ by $P_i(x)$, and $P_m(x)$ is the last non-zero polynomial in the sequence.

Assume that $P(x)$ has no multiple roots. Let $w(a)$ be the number of sign changes in the sequence $P_0(a), P_1(a), \dots, P_m(a)$. Then the number of distinct real roots of $P(x)$ in the open interval $(a, b)$ is equal to $w(a) - w(b)$.

If $P(a) \neq 0$ and $P(b) \neq 0$, then $w(a) - w(b)$ also gives the number of distinct real roots in the closed interval $[a, b]$. If $a$ or $b$ is a root of $P(x)$, the theorem still holds for the open interval.
\end{theorem}
\end{proof}
\end{proof}

\end{document}